\documentclass[%
 reprint,
superscriptaddress,
 amsmath,amssymb,
 aps,
prd
]{revtex4-2}
\usepackage{tikz}
\usepackage{pgfplots}
\pgfplotsset{compat=1.18}
\usepackage{graphicx}% Include figure files
\usepackage{dcolumn}% Align table columns on decimal point
\usepackage{bm}% bold math
\usepackage{subcaption}
\usepackage{mathrsfs}
\usepackage{diffcoeff}
\usepackage[capitalise]{cleveref}
\usepackage{tensor}
\usepackage{amsthm}
\usepackage{nicematrix}
\usepackage{mathtools}

\graphicspath{ {Images/} }

\newcommand{\scri}{\mathscr{I}}
\newcommand{\scrim}{\scri^-}
\newcommand{\scrip}{\scri^+}
\newcommand{\FCG}{\text{fully-compactified gauge}}
\newcommand{\SCG}{\text{semi-compactified gauge}}
\newcommand{\hr}{\hat{r}}
\newcommand{\ii}{\mathrm{i}}

\newcommand{\set}[1]{\left\{#1\right\}}
\newcommand{\swsh}[2]{\tensor[_{#1}]Y{_{#2}}}
\newtheorem{theorem}{Theorem}
\newtheorem{prop}{Proposition}
\newtheorem{lemma}{Lemma}
\newcommand{\ltwo}{\ell^2}

\newcommand{\dd}{\mathrm d}
\renewcommand{\diff}[3][]{%
  \frac{\dd^{#1} #2}{\dd #3^{#1}}%
}

\begin{document}

\preprint{APS/123-QED}

\title{Fully global numerical evolutions of the linearised spin-2 equation}% Force line breaks with \\

\author{Merlyn Barrer}
\email{merlyn.barrer@monash.edu}
\affiliation{School of Mathematics and Statistics, University of Canterbury, Christchurch 8041, New Zealand}
\affiliation{School of Mathematics, Monash University, VIC 3800 Australia}

\author{J\"{o}rg Frauendiener}
\email{joerg.frauendiener@otago.ac.nz}
\affiliation{Department of Mathematics and Statistics, University of Otago,
  Dunedin 9016, New Zealand
}

\author{J\"{o}rg Hennig}
\email{joerg.hennig@otago.ac.nz}
\affiliation{Department of Mathematics and Statistics, University of Otago,
  Dunedin 9016, New Zealand
}

\author{Oliver Markwell}
\email{oliver.markwell@aei.mpg.de}
\affiliation{School of Mathematics and Statistics, University of Canterbury, Christchurch 8041, New Zealand}
\affiliation{Max Planck Institute for Gravitational Physics (Albert Einstein Institute), Am
Mühlenberg 1, Potsdam 14476, Germany}

\author{Chris Stevens}
\email{chris.stevens@canterbury.ac.nz}
\affiliation{School of Mathematics and Statistics, University of Canterbury, Christchurch 8041, New Zealand}

\author{Jed Thompson-Fawcett}
\email{thoje591@student.otago.ac.nz}
\affiliation{Department of Mathematics and Statistics, University of Otago,
  Dunedin 9016, New Zealand
}

\date{\today}

\begin{abstract}
We study the linearised spin-2 field on conformally compactified Minkowski spacetime. Initial data are prescribed on past null infinity $\scrim$, and the system is evolved through the cylinder at spatial infinity to future null infinity $\scrip$.

To avoid the development of logarithmic singularities and thereby preserve smooth peeling behaviour, we derive explicit conditions on the freely specifiable ingoing gravitational radiation $\psi_0$ at $\scrim$ that ensure regular initial data. For compactly supported data, we obtain necessary and sufficient integral conditions guaranteeing regularity and compact support of the full system. 
In the simplest cases, we construct examples of exact initial data, while for the general case, we present a numerical procedure that computes the remaining $\psi_k$ to obtain a full initial data set.

We perform numerical evolutions using two gauges corresponding to compactification of either a neighbourhood of spatial infinity or the entire Minkowski spacetime, and demonstrate stable propagation of the system from $\scrim$ to $\scrip$.

These results provide a step towards a framework for fully global simulations of gravitational radiation and clarify the role of initial data on $\scrim$ in determining asymptotic structure.
\end{abstract}

\maketitle

%\tableofcontents

\section{\label{sec:intro}Introduction}

Understanding the global propagation of gravitational radiation in asymptotically flat spacetimes remains a central problem in general relativity. Of particular interest is the behaviour of fields at past and future null infinity $\mathscr{I}^\pm$, where radiation is unambiguously defined and where key physical quantities, such as energy flux, are extracted. Friedrich’s conformal framework \cite{frauendiener2004conformal} provides a powerful approach to this problem by representing infinity as part of a finite, regular manifold, allowing the study of global properties of solutions within a unified setting.

A numerical implementation of the fully nonlinear Generalised Conformal Field Equations (GCFE) has been developed in recent years \cite{frauendiener2021non,frauendiener2023non,frauendiener2024non}, based on an initial boundary value problem (IBVP) formulation. Although successful in propagating radiation to $\scrip$, and more recently also extending the computational domain to include $\scrim$ \cite{frauendiener2025fully}, this framework remains based on the evolution of \emph{spacelike} hypersurfaces that cross $\scrim$, rather than on data prescribed intrinsically on $\scrim$ itself. However, for the study of scattering problems in asymptotically flat spacetimes, it is desirable to prescribe the free data directly on $\scrim$ and to determine the associated gauge and constrained quantities there, thereby obtaining an initial data set for the Asymptotic Characteristic Initial Value Problem (ACIVP) \cite{kroon2023}. As $\scrim$ is a null hypersurface, the current IBVP implementation---which is based on the evolution of spacelike slices---is not immediately suitable for this purpose.

The construction of a full scattering framework presents several obstacles, the central one being the construction of initial data on $\scrim$. While the conformal structure of asymptotically flat spacetimes renders $\scrim$ accessible at finite location, prescribing data on the entire surface---from past timelike infinity $i^-$ to the bottom of the cylinder at spatial infinity $I^-$---introduces subtle analytical difficulties. Previous work has addressed aspects of this problem in complementary regimes. Friedrich analysed the regularity of fields near timelike infinity in the context of purely radiative spacetimes, establishing conditions under which smooth conformal extensions exist \cite{friedrich1986purely}. On the other hand, the behaviour near the cylinder at spatial infinity has been studied extensively by Kroon et al.\ \cite{taujanskas2023controlled,marajh2025controlled}, who derived conditions ensuring regularity and controlled asymptotics at null infinity, including the avoidance of logarithmic terms that spoil the classical peeling behaviour.

These works give a prescription for controlling the fall-off at each end of $\scrim$ separately. However, the constrained components of the rescaled Weyl curvature on $\scrim$ are governed by a system of first-order equations arising from the Bianchi identity. As a consequence, one cannot freely prescribe the behaviour of the fields \emph{at both ends} of $\scrim$---the regularity at $i^-$ and at $I^-$ are linked in a nontrivial way. Generic initial data therefore lead to polyhomogeneous expansions and the development of logarithmic singularities \cite{kroon2002}.

Even when a suitable initial data set for the ACIVP is obtained, the problem remains of how to evolve these data in a manner compatible with existing numerical frameworks. Doulis and Frauendiener addressed this issue by considering evolutions from a spacelike hypersurface intersecting the middle of the cylinder at spatial infinity \cite{doulis2013second,GlobalSimulationsJoerg}. Working in the linearised regime around Minkowski spacetime, they showed that the spin-2 system decouples from the remaining perturbations, forming a closed, constrained symmetric hyperbolic subsystem. They demonstrated stable numerical evolutions from such initial data through the cylinder and up to $I^+$ (where the cylinder approaches $\scrip$) in a variety of conformal gauges, including those incorporating the physical origin.

The aim of this paper is to address the above issues in the linearised regime around Minkowski spacetime, as a step towards the fully nonlinear case, by developing a framework for fully global numerical evolutions from $\scrim$ to $\scrip$. We derive explicit conditions on the freely specifiable ingoing radiation field that ensure regular initial data and controlled behaviour at the conformal boundary. In particular, we identify necessary and sufficient conditions for compactly supported data to generate regular solutions and analyse how these conditions relate to the asymptotic structure of the spacetime.

In a variety of conformal gauges, our method evolves directly from $\scrim$---in contrast to previous approaches that start from a spacelike initial hypersurface at a finite time. Furthermore, we introduce new techniques for stable evolution through the cylinder at spatial infinity $I$, overcoming difficulties associated with degeneracies of the equations at this boundary.

The paper is organised as follows. In Sec.~\ref{sec:II}, we review conformal compactifications of Minkowski spacetime, introducing both semi-compactified and fully-compactified representations that will be used throughout. We also present the linearised spin-2 system in these settings. In Sec.~\ref{initial data}, we construct initial data on past null infinity $\scrim$. We analyse the intrinsic equations governing the spin-2 field, derive explicit solutions for the remaining components in terms of the freely specifiable data $\psi_0$, and establish conditions ensuring regularity at the cylinder and, where applicable, at past timelike infinity. In particular, we obtain necessary and sufficient conditions for compactly supported data to generate regular initial data sets. We also present a numerical procedure for constructing the full initial data set by computing the remaining $\psi_k$ from a given $\psi_0$. In Sec.~\ref{sec:IV}, we describe the numerical implementation of the evolution system in both semi-compactified and fully-compactified gauges, including our treatment of the cylinder at spatial infinity and the handling of degeneracies in the equations. We then present numerical results demonstrating stable evolutions from $\scrim$ through the cylinder to $\scrip$, and analyse the effect of different initial data choices on the global behaviour of the solutions. Finally, in Sec.~\ref{sec:summary} we summarise our findings and discuss possible extensions to the fully nonlinear case. We use the conventions of \cite{penrose1984spinors} throughout.

\section{Conformal compactifications of Minkowski spacetime}\label{sec:II}
This section describes the construction of several conformally-equivalent representations of Minkowski spacetime that will be used throughout. The usual Minkowski spacetime metric written in spherical coordinates
\begin{equation}\label{eq:MinkPhys}
    \tilde{g} = \dd T^2 - \dd R^2 - R^2(\dd\theta^2 + \sin^2\theta\,\dd\phi^2)
\end{equation}
is used as a common starting point, with different conformal rescalings and coordinate transformations applied in each representation.

\subsection{Semi-compactified representations}
The first representation of interest was developed by Friedrich in \cite{2003CQGra..20..101F,FRIEDRICH199883} and describes Minkowski spacetime in a region around spacelike and null infinity, but does not cover the entire spacetime. Performing the coordinate inversions

\begin{equation*}
R = -\frac{\hr}{\hat t^2 - \hr^2},\quad T = -\frac{\hat t}{\hat t^2 - \hr^2},
\end{equation*}
and compactifying with the conformal factor $\Omega = \hr^2 - \hat t^2$ results in the conformal metric
\begin{equation*}
    \bar{g} = \dd \hat t^2 - \dd\hr^2 - \hr^2(\dd\theta^2 + \sin^2\theta\,\dd\phi^2).
\end{equation*}
This metric is regular through the origin $\hr=0$ that corresponds to spatial infinity. However, it turns out that the linear spin-2 equations and other conformally invariant wave equations are singular there. Furthermore, it was shown in \cite{FRIEDRICH199883} that in general asymptotically flat spacetimes the Weyl curvature also becomes singular near spacelike infinity. In order to resolve this singular behaviour, one introduces an additional rescaling with a function $\kappa(\hr)$ with the property that $\kappa(\hr)/\hr$ is finite at $\hr=0$. With this rescaling, and a new time coordinate $t = \kappa(\hr)^{-1}\hat t$, results in
\begin{equation*}
    g = \dd t^2 + \frac{2t\kappa'\dd t \dd\hr}{\kappa} - \frac{1-t^2\kappa'^2}{\kappa^2}\dd\hr^2 - \frac{\hr^2}{\kappa^2}(\dd\theta^2 + \sin^2\theta\,\dd\phi^2),
\end{equation*}
where $\kappa':=\dd\kappa/\dd\hr$.
The effect of the additional conformal factor is a ``blow up'' of spatial infinity to a cylinder. This metric is related to the original Minkowski metric by the conformal transformation
\begin{equation*}
    g = \Theta^2 \tilde{g}, \quad \Theta = \frac{\hr^2 - \kappa^2 t^2}{\kappa}.
\end{equation*}
For consistency with the other representations used in this paper, the coordinate transformation $\hr = 1 - r$ is performed. This results in the physical spacetime being ``to the left'' of the cylinder at $r=1$ instead of ``to the right'' of the cylinder at $\hr=0$. In these coordinates, the conformal metric is
\begin{equation}
\begin{split}
    g = \dd t^2 + \frac{2 t\kappa'\dd t\dd r}{\kappa} - \frac{1-t^2\kappa'^2}{\kappa^2}\dd r^2&\\
    - \frac{(1 - r)^2}{\kappa^2}(\dd\theta^2 + \sin^2\theta\,\dd\phi^2),
\end{split}
\end{equation}
with adjusted definition $\kappa'=\dd\kappa/\dd r$,
and the conformal factor is
\begin{equation}\label{eq:Omega}
    \Theta = \frac{(1-r)^2 - \kappa^2 t^2}{\kappa}.
\end{equation}
We consider the two choices of $\kappa$
\begin{equation}
    \kappa_{\mathrm{diag}} = \frac{\hr}{1+\hr} = \frac{1 - r}{2 - r},\quad
    \kappa_{\mathrm{hor}} = \hr = 1 - r,
\end{equation}
which define the \emph{diagonal} and \emph{horizontal} semi-compactified representations, respectively, see Fig.~\ref{fig:Semi-compactified}. 
\begin{figure}[ht!]
    \begin{subfigure}{0.25\linewidth}
        \centering
        \begin{tikzpicture}[baseline,remember picture]
        
            \draw (-2,-3) -- node[anchor = west]{$\mathscr{I}^-$} (0,-1);
            \draw (0,1)  -- node[anchor = west]{$\mathscr{I}^+$} (-2,3);
            \draw (0,-1) -- (0,1);
            \draw [dashed] (-2,0) -- (0,0);
            \filldraw (0,1) circle (1pt) node[anchor = west]{$I^+$};
            \filldraw (0,-1) circle (1pt) node[anchor = west]{$I^-$};
            
        \end{tikzpicture}
        \caption{The diagonal representation.}
        \label{cylinderd}
    \end{subfigure}
    \begin{subfigure}[b]{0.25\linewidth}
        \centering
        \begin{tikzpicture}[remember picture,scale = 1]
            \useasboundingbox (-3,-3) rectangle (0,3);
            \draw (0,-1) -- node[anchor = north]{$\mathscr{I}^-$} (-2,-1);
            \draw (-2,1)  -- node[anchor = south]{$\mathscr{I}^+$} (0,1);
            \draw (-0,-1) -- (0,1);
            \draw [dashed] (-2,0) -- (0,0);
            \filldraw (0,1) circle (1pt) node[anchor = west]{$I^+$};
            \filldraw (0,-1) circle (1pt) node[anchor = west]{$I^-$};
        
        \end{tikzpicture}
        \caption{The horizontal representation.}
        \label{cylinderh}
    \end{subfigure}
    \caption{Semi-compactified representations of the cylinder.}\label{fig:Semi-compactified}
\end{figure}
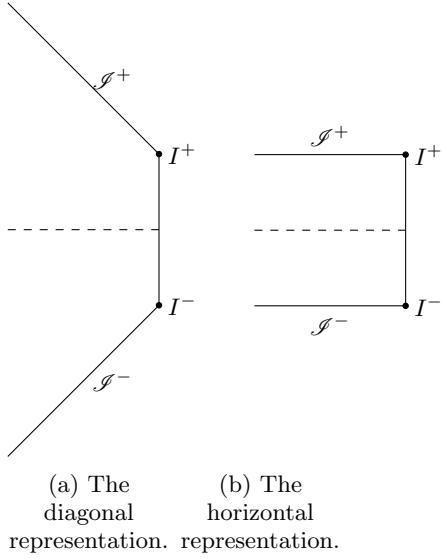

Both bring spatial infinity to a finite location, but they are \textit{semi-compactified} in the sense that they do not include the entire physical spacetime and its conformal boundary. The physical origin $R=0$ and past and future timelike infinity are still infinitely far away in the spatial and temporal directions, respectively.

\subsection{Fully-compactified representations}
This representation, developed in \cite{GlobalSimulationsJoerg}, includes the entire physical spacetime and conformal boundary. This is, of course, advantageous if we want to study \emph{fully global} properties. However, it also introduces additional complexity around the physical origin and timelike infinity, as will be seen in the following discussion.

This representation is achieved through the identification of the entire spacetime, including its conformal boundary, as a submanifold of the Einstein cylinder. Beginning with the Minkowski spacetime Eq.~\eqref{eq:MinkPhys}, performing the coordinate transformation
\begin{equation*}
\begin{split}
    \tau &= \arctan(T + R) + \arctan(T - R),\\
    \rho &= \arctan(T + R) - \arctan(T - R),
\end{split}
\end{equation*}
and compactifying with the conformal factor
\begin{equation*}
    \Omega = 2\cos\left(\frac{\tau + \rho}{2}\right)\cos\left(\frac{\tau - \rho}{2}\right),
\end{equation*}
results in the familiar Einstein cylinder metric
\begin{equation}
    \bar{g} = \dd\tau^2 - \dd\rho^2 - \sin^2\rho\,(\dd\theta^2 + \sin^2\theta\,\dd\phi^2).
\end{equation}
Again, we rescale the conformal factor as $\Omega \rightarrow \kappa^{-1}\Omega$, and introduce new coordinates $(t,r)$ such that $(\tau,\rho) = (\kappa(r)f(t),\pi r)$, to get the final conformal metric
\begin{equation}
\begin{split}
    g = \dot{f}^2\dd t^2 + \frac{2\kappa'f\dot{f}}{\kappa}\dd t\dd r - \frac{\pi^2 - f^2\kappa'^2}{\kappa^2}\dd r^2&\\ 
    - \frac{\sin^2(\pi r)}{\kappa^2}(\dd\theta^2 + \sin^2\theta\,\dd\phi^2)&,
\end{split}
\end{equation}
where $\dot f:=\dd f/\dd t$.
There is a large class of functions $\kappa(r)$ and $f(t)$ that maintain regularity at the origin $r=0$ and at the cylinder $r=1$ \cite{GlobalSimulationsJoerg}. One convenient example is
\begin{equation}
    \kappa(r) = \cos\left(\frac{\pi r}{2}\right),\quad f(t) = 2t.
\end{equation} 
With this choice, the spacetime has spatial infinity represented as the cylinder at $r = 1$, timelike infinity is located at $(t,r) = (\pm\frac{\pi}{2},0)$, and null infinity is given by
\begin{equation}
    t = \pm\frac{\pi(1-r)}{2\cos(\frac{\pi r}{2})}.
\end{equation}
This representation is called the \emph{fully-compactified diagonal representation}, see Fig.~\ref{fig:FullyCompactifiedDiag}.

\begin{figure}[htbp]
    \centering
    \begin{tikzpicture}
        \draw (-0.5,-2) -- (5.5,-2);
        \draw (-0.5,2) -- (5.5,2);
        \draw (-0.5,-2) -- (-0.5,2);
        \draw (5.5,-2) -- (5.5,2);
        \draw[step=1cm,gray,very thin] (-0.5,-2) grid (5,2);
        \draw[very thick] (5,-1) -- (5,1);
        \draw[dashed] (0,-1.57) -- (0, 1.57);
        \draw plot coordinates {(0,1.570796327) (0.5,1.431338851) (1,1.3213064) (1.5,1.234062152) (2,1.164966623) (2.5,1.110720735) (3,1.068959332) (3.5,1.037992862) (4,1.016640738) (4.5,1.004124204) (5,1)};
        \draw plot coordinates {(0,-1.570796327) (0.5,-1.431338851) (1,-1.3213064) (1.5,-1.234062152) (2,-1.164966623) (2.5,-1.110720735) (3,-1.068959332) (3.5,-1.037992862) (4,-1.016640738) (4.5,-1.004124204) (5,-1)};
        \node [anchor = south] at (2.5,1.2) {$\mathscr{I}^+$};
        \node [anchor = north] at (2.5,-1.2) {$\mathscr{I}^-$};
        \filldraw (0,1.57) circle (1pt);
        \node [anchor = north] at (-0.25,2) {$i^+$};
        \filldraw (0,-1.57) circle (1pt);
        \node [anchor = south] at (-0.25,-2) {$i^-$};
        \node [anchor = north] at (0,-2) {0};
        \node [anchor = north] at (1,-2) {0.2};
        \node [anchor = north] at (2,-2) {0.4};
        \node [anchor = north] at (3,-2) {0.6};
        \node [anchor = north] at (4,-2) {0.8};
        \node [anchor = north] at (5,-2) {1};
        \node [anchor = east] at (-0.5,-2) {-2};
        \node [anchor = east] at (-0.5,-1) {-1};
        \node [anchor = east] at (-0.5,0) {0};
        \node [anchor = east] at (-0.5,1) {1};
        \node [anchor = east] at (-0.5,2) {2};
        \node [anchor = west] at (4.9,1) {$I^+$};
        \node [anchor = west] at (4.9,-1) {$I^-$};
        \node [anchor = east] at (-1,0) {$t$};
        \node [anchor = north] at (2.5,-2.5) {$r$};
    \end{tikzpicture}
    
    \caption{The fully-compactified diagonal representation.}\label{fig:FullyCompactifiedDiag}
    
\end{figure}
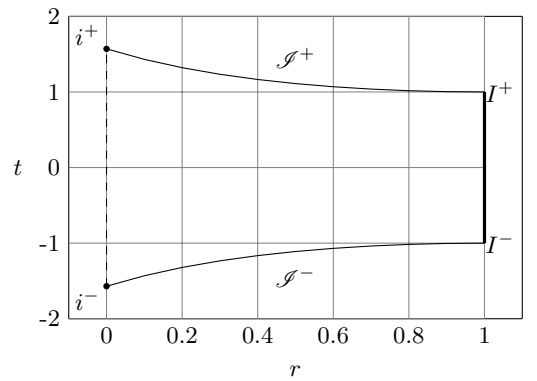
An alternative choice of $\kappa$ and $f$ is
\begin{equation}
    \kappa(r) = \cos\left(\frac{\pi r}{2}\right),\quad f(t) = \frac{1}{20}\tanh^{-1}(t).
\end{equation}
Here, the cylinder is located at $r = 1$, and null infinity is given by
\begin{equation}
    t = \pm \tanh\left(\frac{20\pi(1-r)}{\cos(\frac{\pi r}{2})}\right).
\end{equation}
The factor inside the $\tanh$ was chosen in such a way that the the variation of the right-hand side varies over the entire interval $[0,1]$ is sufficiently small, so that $\scri^{\pm}$ are numerically indistinguishable from horizontal boundaries. Nonetheless, the cone structure at $r=0$ is intrinsic and remains. Consequently, this representation behaves as the horizontal representation while simultaneously compactifying the entire spacetime into $0 \leq r \leq 1$, with timelike infinity located at $i^\pm = (\pm 1, 0)$. We therefore refer to this as the \emph{fully-compactified horizontal representation}, see Fig.~\ref{fig:FullyCompactifiedHoriz}.
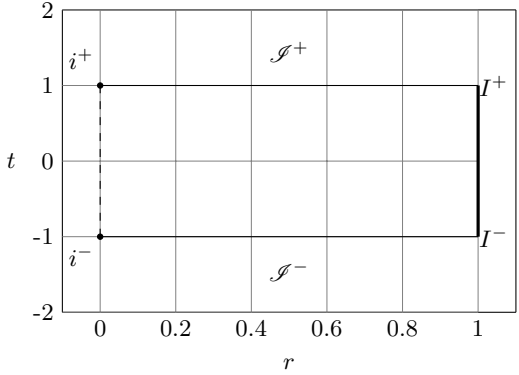
\begin{figure}[h!]
    \centering
    \begin{tikzpicture}
        \draw (-0.5,-2) -- (5.5,-2);
        \draw (-0.5,2) -- (5.5,2);
        \draw (-0.5,-2) -- (-0.5,2);
        \draw (5.5,-2) -- (5.5,2);
        \draw[step=1cm,gray,very thin] (-0.5,-2) grid (5,2);
        \draw[very thick] (5,-1) -- (5,1);
        \draw[dashed] (0,-1) -- (0, 1);
        \draw (0,1) -- (5,1);
        \draw (0,-1) -- (5,-1);
        \node [anchor = south] at (2.5,1.2) {$\mathscr{I}^+$};
        \node [anchor = north] at (2.5,-1.2) {$\mathscr{I}^-$};
        \filldraw (0,1) circle (1pt);
        \node [anchor = north] at (-0.25,1.6) {$i^+$};
        \filldraw (0,-1) circle (1pt);
        \node [anchor = south] at (-0.25,-1.5) {$i^-$};
        \node [anchor = north] at (0,-2) {0};
        \node [anchor = north] at (1,-2) {0.2};
        \node [anchor = north] at (2,-2) {0.4};
        \node [anchor = north] at (3,-2) {0.6};
        \node [anchor = north] at (4,-2) {0.8};
        \node [anchor = north] at (5,-2) {1};
        \node [anchor = east] at (-0.5,-2) {-2};
        \node [anchor = east] at (-0.5,-1) {-1};
        \node [anchor = east] at (-0.5,0) {0};
        \node [anchor = east] at (-0.5,1) {1};
        \node [anchor = east] at (-0.5,2) {2};
        \node [anchor = west] at (4.9,1) {$I^+$};
        \node [anchor = west] at (4.9,-1) {$I^-$};
        \node [anchor = east] at (-1,0) {$t$};
        \node [anchor = north] at (2.5,-2.5) {$r$};
    \end{tikzpicture}
    
    \caption{The fully-compactified horizontal representation.}\label{fig:FullyCompactifiedHoriz}
    
\end{figure}

It is worth noting that unlike the semi-compactified representation, the fully-compactified coordinate system is not adapted to congruence of timelike conformal geodesics required to be a conformal Gau\ss ian system \cite{friedrich1998}. Although spatially constant curves are indeed timelike conformal geodesics, the preferred parameter along the conformal geodesics is not being used as the time coordinate.

\subsection{The linearised equations}

As we are interested in gravitational radiation propagating from past null infinity to future null infinity, the generalised conformal field equations of Helmut Friedrich \cite{friedrich1995einstein} provide a natural framework. We employ these equations in the space-spinor formalism \cite{Beyer_2017}. As a first step, we consider linear perturbations about Minkowski spacetime. In this setting, the perturbation equations for the gravitational spinor $\psi_{ABCD} := \Theta^{-1}\Psi_{ABCD}$ decouple from the remaining perturbations, so that gravitational radiation can be analysed within a symmetric hyperbolic subsystem. We call the eight complex components of this subsystem the \emph{spin-2 equations}. These correspond to the components of the second Bianchi identity in vacuum, $\nabla_{A'}{}^A\psi_{ABCD} = 0$.

We now introduce a null tetrad in order to implement the $\eth$-formalism, an elegant way to circumvent the pole problem: The associated derivative operators act simply on the Spin-Weighted Spherical Harmonic (SWSH) basis functions ${}_sY_{lm}$, and admit a fast and accurate numerical implementation. For a metric of the general form
\begin{equation}\label{eq:genLE}
    g = a^2 \text{d}t^2 - 2c\,\text{d}t\text{d}r - b^2\,\text{d}r^2 - h^2r^2(\text{d}\theta^2 + \sin^2\theta\,\text{d}\phi^2),
\end{equation}
we can define a null tetrad via
\begin{equation}\label{eq:nulltetrad}
    \begin{split}
        l^\mu &= \frac{1}{\sqrt{2}}(A,B,0,0), \\
        n^\mu &= \frac{1}{\sqrt{2}}(C,-B,0,0), \\
        m^\mu &= \frac{1}{\sqrt{2hr}}(0,0,1,\frac{\ii}{\sin\theta}),
    \end{split}
\end{equation}
where
\begin{equation}
    \begin{split}
        A &= \frac{1}{a}\left(1 + \frac{c}{\sqrt{a^2b^2 + c^2}}\right), \\
        B &= \frac{a}{\sqrt{a^2b^2 + c^2}}, \\
        C &= \frac{1}{a}\left(1 - \frac{c}{\sqrt{a^2b^2 + c^2}}\right).
    \end{split}
\end{equation}
The $\eth$ and $\eth'$ derivatives from the GHP-formalism \cite{penrose1984spinors} for this tetrad are implemented via
\begin{equation}
    \begin{split}
         \eth f &= (m^a\nabla_a - p \beta + q\bar\beta')f, \\
         \eth' f &= (\bar m^a\nabla_a + p \beta' - q\bar\beta)f,
    \end{split}
\end{equation}
where $\beta$ and $\beta'$ are spin-coefficients, and $p$ and $q$ are the GHP weights of the weighted quantity $f$. 

At this stage, it is important to note that Eq.~\eqref{eq:genLE} describes \emph{round} 2-spheres for constant $t$ and $r$. As there exist fast and accurate methods to implement the $\eth$-operators numerically for the unit 2-sphere \cite{beyer2016}, denoted $\hat\eth$ and $\hat\eth'$, we will utilise the simple relations
\begin{equation}
    \eth = \frac{1}{\sqrt{2}hr}\hat\eth, \qquad
    \eth' = \frac{1}{\sqrt{2}hr}\hat\eth'.
\end{equation}

To exploit the linearity and spherical nature of the setting, the components of the gravitational spinor are expanded in the SWSH basis ${}_sY_{lm}$ \cite{beyer2016}. Namely, for a type-$\{p,q\}$ function $f$ \emph{\`a la}  \cite{penrose1984spinors}, which has spin-weight $s=(p-q)/2$ and is $L^2$-integrable over the 2-sphere, we have the representation
\begin{equation}
    f = \sum_{l=|s|}^\infty\sum_{m=-l}^l a_{lm}\;{}_sY_{lm},
\end{equation}
where $a_{lm}$ are the spectral coefficients for $f$, which are constant on the sphere.

The $\hat\eth$-operators act on the SWSH basis functions via
\begin{equation}
    \begin{split}\label{eq:ethaction}
        \hat\eth\,{}_sY_{lm} = -\sqrt{l(l+1)-s(s+1)}\,{}_{s+1}Y_{lm}, \\
        \hat\eth'\,{}_sY_{lm} = \sqrt{l(l+1)-s(s-1)}\,{}_{s-1}Y_{lm}.
    \end{split}
\end{equation}
An important consequence of the linearity of the spin-2 equations is that each term in an SWSH expansion evolves independently, i.e.\ there is no mode mixing. We therefore formulate the equations for general SWSH modes $(l,m)$, with the restriction $l \ge 2$ appropriate for radiative (spin-2) degrees of freedom.
We adopt a slightly informal notation and use $\psi_k$ both for the full components of the spin-2 field and for their individual $(l,m)$ modes, with the meaning being clear from context.

The eight linearised equations for the $(l,m)$ modes split into five evolution equations and three constraint equations. In the semi-compactified regime, these take the form \cite{doulis2013second}
\begin{equation} \label{ethevo}
    \begin{split}
        (1 + t \kappa') \partial_t \psi_0 &= \kappa\partial_r\psi_0 - \left(3\kappa' +\frac{\kappa}{1-r}\right)\psi_0 - \frac{\kappa}{1-r}a_0\psi_1 ,\\
        \partial_t \psi_1 &= \frac{\kappa}{2(1-r)}\left(a_0\psi_0+2\psi_1-a_1\psi_2\right),\\
        \partial_t \psi_2 &= \frac{\kappa}{2(1-r)}(a_1\psi_1 - a_2\psi_3),\\
        \partial_t \psi_3 &= \frac{\kappa}{2(1-r)}\left(a_2\psi_2-2\psi_3-a_3\psi_4\right),\\
        (1 - t\kappa') \partial_t \psi_4 &= -\kappa\partial_r\psi_4 + \frac{\kappa}{1-r}a_3\psi_3+\left(3\kappa' +\frac{\kappa}{1-r}\right)\psi_4 ,
    \end{split}
\end{equation}
and
\begin{equation} \label{ethconstr}
    \begin{split}
        2(1-r)\partial_r\psi_1 &= a_0(1+t\kappa')\psi_0+2\left[3 + \left(3\frac{1-r}{\kappa}+t\right)\kappa'\right]\psi_1 \\
        & \quad\quad +a_1(1-t\kappa')\psi_2,\\
        2(1-r)\partial_r\psi_2 &= a_1(1+t\kappa')\psi_1 +6\left(1+\frac{1-r}{\kappa}\kappa'\right)\psi_2\\
        & \quad\quad +a_2(1-t\kappa')\psi_3,\\
        2(1-r)\partial_r\psi_3 &=a_2(1+t\kappa')\psi_2 +2\left[3+\left(3\frac{1-r}{\kappa}-t\right)\kappa'\right]\psi_3\\
        & \quad\quad +a_3(1-t\kappa')\psi_4,
    \end{split}
\end{equation}
where
\begin{equation}\label{eq:akDefn}
    a_k = \sqrt{l(l+1)-(k-1)(k-2)}.
\end{equation}
The coefficient $a_k$ arises upon acting on a spin-weighted spherical harmonic $\swsh{2-k}{lm}$ with the $\eth'$ operator through Eq.~\eqref{eq:ethaction}.

Note that different choices of the coordinate function $\kappa(r)$ in these equations are completely equivalent, 
as the corresponding spacetimes differ only by a coordinate transformation combined with a conformal transformation, the latter arising because our choice of conformal factor \eqref{eq:Omega} depends explicitly on $\kappa$. Consequently, if $\psi_k$ and $\tilde\psi_k$ represent the same solution of the spin-2 equations
expressed in coordinates associated with $\kappa$ and $\tilde\kappa$, respectively, we obtain the transformation law
\begin{equation}\label{eq:psitrafo}
    \tilde\psi_k(\tilde t, r) = \frac{\tilde\kappa^3(r)}{\kappa^3(r)} 
    \psi_k\!\left(\frac{\tilde\kappa(r)}{\kappa(r)}\,\tilde t, r\right).
\end{equation}
This relation can also be verified directly by checking that, if $\psi_k$ satisfies the constraint
and evolution equations for a given function $\kappa$, then $\tilde\psi_k$ defined by
\eqref{eq:psitrafo} satisfies the corresponding equations with $\tilde\kappa$.
As a result, most analytical considerations may be restricted to the simplest choice
$\kappa=\kappa_{\mathrm{hor}}=1-r$, with all results readily transferred to arbitrary coordinates
associated with a general function $\tilde\kappa$.

On the other hand, in the fully-compactified regime, defining
\begin{equation}
    \begin{split}
        A = \frac{1}{\dot{f}(t)}&\left(1 + \frac{f(t)}{2}\sin\left(\frac{\pi r}{2}\right)\right),\quad B = \frac{1}{\pi}\cos\left(\frac{\pi r}{2}\right),\\&\quad C = \frac{1}{\dot{f}(t)}\left(1 - \frac{f(t)}{2}\sin\left(\frac{\pi r}{2}\right)\right),\\
            &h %= \frac{\sin(\pi r)}{r\cos(\frac{\pi r}{2})},
                = \frac{2}{r}\sin\left(\frac{\pi r}{2}\right),
            \quad \epsilon = \frac{1}{4\sqrt{2}}\sin\left(\frac{\pi r}{2}\right),\quad \\ &\rho = -\frac{\sin(\frac{\pi r }{2}) + 2\cos(\frac{\pi r}{2})\cot(\pi r)}{2\sqrt{2}},
    \end{split}
\end{equation}
where $\epsilon$ and $\rho$ are spin-coefficients, the evolution and constraint equations are \cite{GlobalSimulationsJoerg}
\begin{equation}\label{compactevo}
    \begin{split}
        C\partial_t\psi_0 &= B\partial_r\psi_0 + \sqrt{2}\left(4\epsilon - \rho\right)\psi_0 - \frac{a_0}{hr}\psi_1,\\
        (A+C)\partial_t\psi_1 &= \frac{a_0}{hr}\psi_0 + 2\sqrt{2}\left(2\epsilon + \rho\right)\psi_1 - \frac{a_1}{hr}\psi_2,\\
        (A+C)\partial_t\psi_2 &= \frac{a_2}{hr}(\psi_1 - \psi_3),\\
        (A+C)\partial_t\psi_3 &= \frac{a_2}{hr}\psi_2-2\sqrt{2}(2\epsilon+\rho)\psi_3 - \frac{a_3}{hr}\psi_4,\\
        A\partial_t\psi_4 &= - B\partial_r\psi_4 + \frac{a_3}{gr}\psi_3 - \sqrt{2}\left(4\epsilon-\rho\right)\psi_4 ,
    \end{split}
    \end{equation}
and
\begin{equation}\label{compactconstraint}
    \begin{split}
        &2(A+C)B\partial_r\psi_1 - \frac{2a_0}{hr}C\psi_0 \\ &
           \qquad + 4\sqrt{2}\left[A(\epsilon - \rho)- C(2\rho + \epsilon)\right]\psi_1 - \frac{2a_1}{hr}A\psi_2 = 0, \\
        &2(A+C)B\partial_r\psi_2 - \frac{2a_2}{hr}(C\psi_1 + A\psi_3)\\ &
          \qquad - 6\sqrt{2}(A+C)\rho\psi_2 = 0,\\
        &2(A+C)B\partial_r\psi_3 - \frac{2a_2}{hr}C\psi_2 \\&  
        \qquad +4\sqrt{2}\left[C(\epsilon - \rho) - A(2\rho + \epsilon)\right]\psi_3 - \frac{2a_3}{hr}A\psi_4 = 0,
    \end{split}
\end{equation}
respectively.

\section{Initial data on $\scri^-$}\label{initial data}

\subsection{The equations and issues to overcome}

In vacuum asymptotically flat spacetimes, the physical degrees of freedom of the spin-2 field are entirely represented by the complex-valued function $\psi_0$ on $\scrim$---the ingoing gravitational radiation. Before we can evolve the spin-2 system beyond this surface, we must first fix all $\psi_k$ components there in a manner consistent with the spin-2 equations.
Recall that the linearised equations consist of five evolution equations and three constraint equations (i.e.\ Eqs.\ \eqref{ethevo} and \eqref{ethconstr} or Eqs.\ \eqref{compactevo} and \eqref{compactconstraint}). 
If we decompose the spin-2 equations with respect to a null vector field adapted to $\scrim$ instead of a timelike vector field, the eight components split into four intrinsic to $\scrim$ and four extrinsic components.
Using the null tetrad \eqref{eq:nulltetrad}, the four intrinsic equations, together with a given $\psi_0$, form a system determining the remaining $\psi_k$ on $\scrim$, which constitute the initial data set.
These equations are
\begin{equation} \label{eq:general_lin}
    A \partial _t \psi_k + B \partial _r \psi_k - \sqrt{2} \bigl( (5-k) \rho + 2(2-k) \varepsilon \bigr) \psi_k = 2 a_{k-1} \psi_{k-1},
\end{equation}
where $k \in \set{1,2,3,4}$, and $a_k$ is given in Eq.~\eqref{eq:akDefn}.
%\begin{equation}\label{eq:akDefn}
%    a_k = \sqrt{6 - (k-2)(k-1)}.
%\end{equation}
% \begin{equation*}
%     a_k = \sqrt{l(l+1) - (k-2)(k-1)}.
% \end{equation*}

Since we consider a linearisation about Minkowski spacetime, $\scrim$ can be easily described in terms of the background's $(t, r)$ coordinates.
We will begin by looking at the semi-compactified gauge, where
\begin{equation}
    \scrim = \set{t = - \frac{\hat r}{\kappa(\hat r)}}.
\end{equation}
In this gauge, the equations \eqref{eq:general_lin} simplify to 
\begin{equation}
    \frac{1+t \kappa'}{\kappa}\partial_t \psi_k - \partial_{\hat r} \psi_k = \left(\frac{5-k}{\hat r} - \frac{3 \kappa'}{\kappa}\right) \psi_k + \frac{a_{k-1}}{\hat r} \psi_{k-1}.
\end{equation}
$\scri$ is, by construction, a characteristic surface of each of these equations, as can also be seen from a direct computation. If $\hat r$ is chosen as a parameter along $\scri$, the equations reduce to a system of ODEs along $\scri$.
In the semi-compactified horizontal representation, where $\kappa(\hat r) = \hat r$, this yields
\begin{equation} \label{eq:semi scrim eqns}
	\diff{\psi_k}{\hat r} =  \frac{k-2}{\hat r} \psi_k - \frac{a_{k-1}}{\hat r} \psi_{k-1}, \quad k \in \set{1,2,3,4}.
\end{equation}

Applying the same process to both fully-compactified representations gives
\begin{equation} \label{eq:main}
	\diff{\psi_k}{r} + \frac{\pi\bigl(3 + (7 - 2 k) \cos(\pi r)\bigr)}{2 \sin(\pi r)} \psi_k = \frac{\pi a_{k-1} }{\sin(\pi r)} \psi_{k-1}.
\end{equation}

Regardless of the gauge choice above, the ODEs are of the form
\begin{equation}
    \psi_k' = - f_k \psi_k + g_k,
\end{equation}
for some functions $f_k$ and $g_k$. 
These can be solved with integrating factors $\mu_k = e^{\int f_k(r) dr}$ to obtain the general solutions.
In the semi-compactified horizontal representation, this yields
\begin{equation} \label{eq:horiz_integral}
	\psi_k(\hr) = - a_k \hr^{k-2} \int_0^{\hat r} \rho^{1-k} \psi_{k-1}(\rho) \dl \rho + C_k \hr^{k-2}, 
\end{equation}
while in the fully-compactified representations, we get
\begin{multline} \label{eq:integral}
    \psi_k(r) = \frac{\pi a_{k-1}}{2} s(r)^{k-5} c(r)^{k-2} 
    \\
    \times \left(\int_0^r s(\rho)^{4-k} c(\rho)^{1-k} \psi_{k-1}(\rho) \dl \rho + C_k\right),
\end{multline}
where
\begin{equation*}
    s(r) = \sin\left(\frac{\pi r}{2}\right), \quad c(r) = \cos\left(\frac{\pi r}{2}\right),
\end{equation*}
and the $C_k$ are integration constants.

The intrinsic equations possess a natural hierarchical structure. Once the freely specifiable ingoing radiation field $\psi_0$ is prescribed, the remaining components are determined sequentially, with each equation involving only $\psi_k$ and the previously determined component $\psi_{k-1}$. This ladder-like structure is closely related to that underlying the Teukolsky--Starobinsky identities on type-D backgrounds \cite{starobinskii1974amplification,teukolsky1974perturbations,kalnins1989teukolsky}, where repeated use of the first-order Bianchi equations yields direct relations between the spin $\pm2$ Weyl scalars. In the present characteristic setting, we instead exploit the hierarchy to construct the remaining components by successive integration along $\scrim$.

Studying the solutions in integral form is helpful for determining conditions on the choice of ingoing radiation $\psi_0$ for which the full initial data set (i.e.\ all $\psi_k$) is regular on all or part of $\scrim$.
It is of direct physical relevance if the initial data set is not regular at past timelike infinity $i^-$ or at the bottom of the cylinder $I^-$. 

Non-regularity at $i^-$ signals that a non-decaying Coulomb-type field (i.e.\ rest mass yielding a non-vanishing $\psi_2$) is already present in the infinite past sourcing $\psi_{ABCD}$. This is of interest as black hole spacetimes fall within this class. In this case, $i^-$ is not a part of a smooth conformal boundary, and some if not all of the $\psi_k$ diverge there. 

Regularity at $I^-$ requires precise cancellations in the asymptotic expansions of data near spacelike infinity. Generic data lead to polyhomogeneous (logarithmic) terms in the $\psi_k$, which propagate to future null infinity and destroy its smoothness, thereby violating the classical peeling behaviour. Physically, this reflects the presence of a persistent long-range gravitational field that cannot be cleanly separated from radiation, rather than a purely transient radiative process.

Hence, investigating properties of initial data sets on $\scrim$ requires us to look closely at regularity of the $\psi_k$ at one or both ends of $\scrim$, namely $i^-$ and $I^-$. It is important to note that because the spin-2 equations intrinsic to $\scrim$ are governed by \emph{first order} equations, we cannot enforce regularity at both ends of $\scrim$ simultaneously. In the following sections, we investigate these questions.

\subsection{Regularity conditions on $\scrim$ in the semi-compactified representation}
\label{Regularity SCG}
In this section, we look for conditions on the ingoing wave $\psi_0$ so that the other components $\psi_k$, as given by the integral \eqref{eq:horiz_integral}, do not diverge anywhere. This corresponds to a completely non-radiative, zero-mass past.

In the semi-compactified representation, $I^-$, the bottom of the cylinder, is located at $\hat r = 0$, and a piece of $\scrim$ in a neighbourhood of $I^{-} $ is parametrised by $0 \leq \hat r < \infty$. 
Because this representation does not compactify to past timelike infinity, we cannot access all of $\scrim$. Instead, we will look for initial data sets that cover a finite portion away from the cylinder. We will set $\psi_0$ to have compact support $[b, a]$ for some $0 \leq b \leq a$, and require it to be continuously differentiable.
%We will then look for initial data sets that have compact support $[0, a)$, and which are continuous and differentiable everywhere, as well as analytic in a neighbourhood of $\hat r=0$. 
% This is illustrated in \cref{fig:sc}.
% \begin{figure}[htbp]
% 	\centering
% 	\begin{subfigure}{0.3\textwidth}
% 		\begin{tikzpicture}
% 			\centering
% 			% Draw scrim
% 			\draw[thick, ->] (0,-1) -- (-2,-3) node[midway, left] {$\scrim$};
		
% 			\draw (0, -1) node[right] {$I^-$ ($r = 0$)};
	
% 			\draw[thick, DodgerBlue] (-0.2, -1.2) -- (-1, -2) node[midway, right, DodgerBlue] {$\psi_0 \neq 0$};
	
% 			\filldraw[DodgerBlue] (-0.2, -1.2) circle (1pt) node[DodgerBlue, right]{$b$};
% 			\filldraw[DodgerBlue] (-1, -2) circle (1pt) node[DodgerBlue, right]{$a$};

% 			% Draw the cylinder
% 			\draw[thick] (0, -1) -- (0, 1) node[midway, right] {$I$};
% 		\end{tikzpicture}
% 	\end{subfigure}
% 	\begin{subfigure}{0.3\textwidth}
% 		\centering
% 		\begin{tikzpicture}
% 			% Draw scrim
% 			\draw[thick, ->] (0,-1) -- (-2,-3) node[midway, left] {$\scrim$};
		
% 			\draw (0, -1) node[right] {$I^-$ ($r = 0$)};
	
% 			\draw[thick, DodgerBlue] (0, -1) -- (-1, -2) node[midway, right, DodgerBlue] {$\psi_k \neq 0$};
	
% 			\filldraw[DodgerBlue] (0, -1) circle (1pt) {};
% 			\filldraw[DodgerBlue] (-1, -2) circle (1pt) node[DodgerBlue, right]{$a$};

% 			% Draw the cylinder
% 			\draw[thick] (0, -1) -- (0, 1) node[midway, right] {$I$};
% 		\end{tikzpicture}
% 	\end{subfigure}
% 	\caption{Schematic illustrating the region where $\psi_k$ may be non-zero.}
% 	\label{fig:sc}
% \end{figure} 

We begin by noting some conditions that are required for regularity at $\hat r = 0$.
Note that these conditions will be trivially satisfied in the case $b > 0$.  
\begin{prop}  \label{prop:r0 conds}
	In the semi-compactified representation, for all $\psi_k(\hat{r})$ to be continuously differentiable at $\hat r = 0$, we require
	\begin{subequations} \label{eq:sc cylinder conds}
		\begin{align}
			\psi_1(0) = \psi_0(0) = 0,
			\\
			\psi_3(0) = a_2 \psi_2(0),
			\\
			\psi_4(0) = \frac{a_3}{2} \psi_3(0),
		\end{align}
	\end{subequations}
    where the $a_k$ are given by Eq.~\eqref{eq:akDefn}.
    Furthermore, if we assume that each $\psi_k$ is twice continuously differentiable at $\hat{r} = 0$, then $\psi_0 = \mathcal{O}(\hat{r}^2)$ and $\psi_1 = \mathcal{O}(\hat{r}^2)$. Similarly, if each $\psi_k$ is thrice continuously differentiable at $\hat{r} = 0$, then $\psi_0 = \mathcal{O}(\hat{r}^3)$ and $\psi_1 = \mathcal{O}(\hat{r}^3)$.
\end{prop}
\begin{proof}
	The entire right-hand side of \cref{eq:semi scrim eqns} has a factor of $1/\hat r$, which is singular at $\hat r = 0$. 
	Thus, for $\difs{\psi _k}{\hat r}$ to be finite, we require that
	\begin{equation}
		(k-2) \psi_k - a_{k-1} \psi_{k-1} = \mathcal{O}(\hat r)
	\end{equation}
	around $\hat r = 0$. This implies
	\begin{equation}
		(k-2) \psi_k(0) - a_{k-1} \psi_{k-1}(0) = 0.
	\end{equation}
	Evaluating this equation for each $k \in \set{1, 2, 3, 4}$, we obtain \cref{eq:sc cylinder conds}.
    
    For the second and third statements, we start by rearranging Eq.~\eqref{eq:semi scrim eqns} to solve for $\psi_{k-1}$. Doing this yields
	\begin{equation}
	\label{eq:semi scrim rearranged}
	\psi_{k-1} = \frac{1}{a_{k-1}} \left ( (k-2)\psi_k - \hat{r} \diff{\psi_k}{\hat r} \right ).
	\end{equation}
    
	If we assume that all $\psi_k$ are twice continuously differentiable and solve for $\psi_1$ in terms of $\psi_3$, we find that
	\begin{align*}
	\psi_1 & = \frac{\hat{r}^2}{a_1a_2}\diff[2]{\psi_3}{\hat{r}}.
	\end{align*}
	Thus, since $\psi_3$ is twice continuously differentiable, we find $\psi_1 = \mathcal{O}(\hat{r}^2)$. Then, as $\psi_1$ is twice continuously differentiable, using a L'H\^opital's rule argument yields that d${\psi_k}/$d$\hat{r} = \mathcal{O}(\hat{r})$. Hence, using this and Eq.~\eqref{eq:semi scrim rearranged} for $\psi_0$, we obtain
	\begin{align*}
	\psi_0 & = -\frac{1}{a_0} \left (\psi_1 + \hat r \diff{\psi_1}{\hat r} \right ) = \mathcal{O}(\hat{r}^2). 
	\end{align*}
	
	Assuming all $\psi_k$ are thrice continuously differentiable yields 
	\begin{align*}
	\psi_1 & = -\frac{\hat{r}^3}{a_1a_2a_3}\diff[3]{\psi_4}{\hat{r}} = \mathcal{O}(\hat{r}^3).
	\end{align*}
	With a similar line of reasoning as in the previous case, we can conclude $\psi_0 = \mathcal{O}(\hat{r}^3)$. 
\end{proof}

In general, a solution with $\psi_k(\hat r) = 0$ for $\hat r \geq a$ may not be regular at $\hat r = 0$. Hence, we want to determine requirements on $\psi_0$ such that all $\psi_k$ are differentiable, and potentially vanishing there whilst also having compact support.  
To initiate this discussion, we temporarily move away from initial data with compact support on $\scrim$ and focus on ensuring regularity for general initial data. 

The ODEs we are solving are linear, meaning we can use a simple integrating factor to find a solution. However, care is needed when we do this, since the equations are singular at $\hat r = 0$. We start by constructing our solutions without any constants of integration and integrate from $0$ to an arbitrary $\hat{r}$, as this proves most fruitful later on. Thus, from Eq.~\eqref{eq:horiz_integral}, these solutions take the form
\begin{align} 
\psi_k(\hat r) = -a_{k-1} \hat r^{k-2} \int_0^{\hat r} \rho^{1-k}\psi_{k-1}(\rho) \dl \rho. \label{eq:horiz_int_2}
\end{align}
Obviously, integrating from $\rho=0$ only makes sense if we can ensure that the integrand is integrable there, which is not trivial, given that for $k=2,3,4$, we have a factor of $1/\rho$, $1/\rho^2$, and $1/\rho^3$ in the integrand, respectively. The information provided in Proposition \ref{prop:r0 conds} helps with this problem, as the nature of Eq.~\eqref{eq:horiz_int_2} has a form of `decay preservation' at $\hat{r} = 0$ when going from $\psi_{k-1}$ to $\psi_k$, given the decay of $\psi_0$ is sufficiently fast. This is summarised in the following. 

\begin{lemma}\label{lem:nice sols}
Let $\psi_0$ be continuously differentiable for $\hat{r}\ge 0$ and $\psi_0 = \mathcal{O}(\hat r^n)$ with $n\ge 3$. Then each $\psi_k$ for $k=1,2,3,4$ defined by Eq.~\eqref{eq:horiz_int_2} is well defined, continuously differentiable, and $\psi_k = \mathcal{O}(\hat r^n)$.
\end{lemma}

\begin{proof}
For this proof, we prove a general case and use the fact that $\psi_0 = \mathcal{O}(\hat{r}^n)$ for $n\ge 3$. Starting with Eq.~\eqref{eq:horiz_int_2}, if $\psi_{k-1} = \mathcal{O}(\rho^{n})$ for $n\ge k-1$, then we can conclude that $\rho^{1-k}\psi_{k-1} $ is Riemann-integrable from $0$ to $\hat r$ for any $\hat r>0$. This is because it can only be discontinuous at $\rho=0$, due to our assumptions, and near $\rho=0$, we have that $ \rho^{1-k}\psi_{k-1}  = \mathcal{O}(\rho^{n-k+1})$ for $n-k+1\ge 0$, meaning it is bounded near its only possible point of discontinuity. Thus, as we know the integral is well defined,
\begin{align*}
\psi_k & = -a_{k-1} \hat r^{k-2} \int_0^{\hat r} \rho^{1-k}\psi_{k-1}  \dl \rho ,
\\
& = -a_{k-1}\hat r^{k-2} \int_0^{\hat r} \mathcal{O}( \rho ^{n-k+1}) \dl \rho ,
\\ 
& = \hat r^{k-2} \mathcal{O}(\hat r ^{n-k+2}) = \mathcal{O}(\hat r^n).
\end{align*}
Using this, and noting that $\psi_0 = \mathcal{O}(\hr^n)$ for $n\ge 3 \ge k-1$ for all $k=1,2,3,4$, we get that all $\psi_k$ are well defined and $\psi_k = \mathcal{O}(\hr^n)$.

Finally, to prove that each $\psi_k$ is continuously differentiable (where differentiation at $\hr = 0$ is taken as $\hr \to 0^+$), we apply the product rule and then the Fundamental Theorem of Calculus. For that purpose, our integrand must be continuous over the interval, which is not guaranteed at $\rho =0$. However, this can be handled by splitting the integral into two parts: one from 0 to 1 and another from 1 to $\hr$ and differentiating them separately, noting that the first integral is just a constant. Then, we can observe that the first derivative of each $\psi_k$ is continuous by looking at the nature of Eq.~\eqref{eq:semi scrim eqns} and noticing that both $\psi_k/\hat r, \, \psi_{k-1}/\hat r \to 0$ as $\hat r \to 0$ for $k=1,2,3,4$.
\end{proof}

We can also prove that the result of Lemma~\ref{lem:nice sols} also holds for $\psi_0 = \mathcal{O}(\hat r^2)$, but it comes with the caveat of defining $\psi_4$ more carefully. This can be seen via the order-preservation property, as proven above, since the integrand in the definition of $\psi_4$ behaves like $\psi_3/\rho^3 = \mathcal{O}(1/\rho)$ near $\rho = 0$. Thus, when defining the integral properly, we find that, in general, $\psi_4 = \mathcal{O}(\hat r^2 \ln \hat r )$ near $\hat r =0$. 

A nice consequence of constructing initial data via the recommendations of Lemma~\ref{lem:nice sols} is that if $\psi_0$ is $m$-times continuously differentiable, then each $\psi_k$ is also $m$-times continuously differentiable. The proof is somewhat lengthy, so we omit it here; however, this straightforward result relies on the fact that if $f(x)/x\to L\in \mathbb{R}$ when $x\to 0$ and $f(x)$ is $(m+1)$-times differentiable on an interval including zero, then $f(x)/x$ is $m$-times differentiable. However, when $\psi_0 = \mathcal{O}(\hat r^2)$, this does not hold, as Proposition \ref{prop:r0 conds} directly disproves this possibility. 

The next step is to include constants of integration into our initial data functions $\psi_k$. As will be seen later, to include these, we will separate the `core' solution structure given in Eq.~\eqref{eq:horiz_int_2} from them. We define $\bar{\psi}_k$ by Eq.~\eqref{eq:horiz_int_2} and incorporate the constants of integration via functions $\beta_k$ for $k=1,2,3,4$. The general solution to Eq.~\eqref{eq:semi scrim eqns} is then given by
\begin{equation} \label{eq:general_sols}
\psi_k = \bar{\psi}_k + \beta_k, 
\end{equation}
where, by linearity of \eqref{eq:semi scrim eqns}, this defines a solution provided that each $\beta_k$ itself is a solution.

As we do not want to introduce any singularities, we must choose these constants wisely. From the solution structure in Eq.~\eqref{eq:horiz_integral}, we do not include an integration constant for $\psi_1$. Otherwise, we would have a term $C_1/\hat r$, which is singular as we approach the cylinder for $C_1 \neq 0$. Hence, $\beta_1=0 $ is enforced for regularity. However, for $k=2,3,4$, the choice of $C_k$ does not affect the regularity of $\psi_k$. Thus, we have the freedom to choose $\beta_k$ for $k=2,3,4$ that satisfy Eq.~\eqref{eq:semi scrim eqns}. It turns out that for any choice of constants $C_k$ for $k=2,3,4$, our functions $\beta_k$ must be 
\begin{align}
\begin{split}
\beta_1 & = 0, 
\\
\beta_2 & = C_2,
\\
\beta_3 & = C_3 \hr + a_2C_2,
\\
\beta_4 & = C_4 \hr^2 + a_3C_3 \hr + \frac{a_3a_2}{2} C_2. 
\end{split}\label{eq:int_consts}
\end{align}

Before returning to solutions with compact support, we note one further useful property of solutions to Eq.~\eqref{eq:semi scrim eqns}. Specifically, the iterated integrals appearing in Eq.~\eqref{eq:horiz_integral} and Eq.~\eqref{eq:horiz_int_2} can be eliminated, allowing each $\psi_k$ to be expressed solely in terms of integrals of $\psi_0$. As we now demonstrate, this is achieved via integration by parts.

\begin{theorem} \label{thm:alt_sol_form}
Given that $\psi_0 = \mathcal{O}(\hat r^n)$ for $n\ge 3$ and continuously differentiable for $\hat r \ge 0$, the general solutions to Eq.~\eqref{eq:semi scrim eqns} for our initial data can be expressed as
\begin{align} \label{eq:alt_sol}
\begin{split}
\psi_1 & = - \frac{a_0}{\hat r} \int_0^{\hat r} \psi_0 \dl \rho,
\\
\psi_2 & = a_1a_0 \left (\int_0^{\hat r} \frac{\psi_0}{\rho}\dl \rho -  \frac{1}{\hat r}\int_0^{\hat r} \psi_0 \dl \rho\right ) + C_2,
\\
\psi_3 & = -\frac{a_2a_1a_0}{2} \left (\hr\int_0^{\hat r}  \frac{\psi_0}{\rho^2}\dl \rho -2 \int_0^{\hat r} \frac{\psi_0}{\rho} \dl \rho +  \frac{1}{\hat r}\int_0^{\hat r} \psi_0 \dl \rho\right )
\\
&  \qquad + C_3 \hat r + a_2 C_2,
\\
\psi_4 & = \frac{a_3a_2a_1a_0}{6} \left (\hr^2\int_0^{\hat r}  \frac{\psi_0}{\rho^3}\dl \rho -3\hr \int_0^{\hat r}  \frac{\psi_0}{\rho^2}\dl \rho +3 \int_0^{\hat r} \frac{\psi_0}{\rho} \dl \rho\right.
\\
& \qquad \left. -  \frac{1}{\hat r}\int_0^{\hat r} \psi_0 \dl \rho\right ) + C_4 \hat r^2+  a_3C_3 \hat r + \frac{a_3 a_2}{2}C_2.
\end{split}
\end{align}
\end{theorem}

\begin{proof}
It suffices to show the equivalence of Eqs.~\eqref{eq:horiz_int_2} and \eqref{eq:alt_sol} for $C_2=C_3=C_4=0$. We begin this proof by considering the various integrals of $\psi_0$ and writing them in terms of the other $\psi_k$ using the formula given in Eq.~\eqref{eq:horiz_int_2}. Before we proceed, note that all integrals in this proof are well-defined because each $\psi_k = \mathcal{O}(\hat r^3)$. This fact will also be used implicitly when evaluating various $\psi_k/\rho$ and $\psi_k/\rho^2$ terms at $\rho =0$. Trivially, we are given the $\psi_1$ relation, as nothing has been augmented. Next, we consider 
the integral of $\psi_0/\rho$.
To employ integration by parts, we use that, according to Eq.~\eqref{eq:horiz_int_2}, an antiderivative of $\psi_0$ is $- \rho \psi_1 /a_0$. Thus, 
\begin{align}
\begin{split}
\int_0^{\hat r} \frac{\psi_0}{\rho} \dl \rho & = - \frac{1}{a_0} \psi_1(\rho) \biggl |_0^{\hat r} - \frac{1}{a_0} \int_0^{\hat r} \frac{\psi_1}{\rho} \dl \rho ,
\\
& = -\frac{1}{a_0} \psi_1 + \frac{1}{a_0a_1} \psi_2,
\end{split}\label{eq:alt_int_1}
\end{align}
noting the use of Eq.~\eqref{eq:horiz_int_2} again to replace the $\psi_1$ integral in terms of $\psi_2$. Applying the same procedure to the next integral gives the following when employing integration by parts while integrating $\psi_0$,
\begin{align*}
\int_0^{\hat r} \frac{\psi_0}{\rho^2} \dl \rho & = -\frac{1}{a_0}\frac{\psi_1}{\hat r} -\frac{2}{a_0}\int_0^{\hat r} \frac{\psi_1}{\rho^2} \dl \rho. 
\end{align*}
We now utilise integration by parts again using the antiderivative $- \psi_2/a_1$ of $\psi_1/\rho$ from Eq.~\eqref{eq:horiz_int_2}. This unveils the following
\begin{align}
\begin{split}
\int_0^{\hat r} \frac{\psi_0}{\rho^2} \dl \rho & = -\frac{1}{a_0}\frac{\psi_1}{\hat r} + \frac{2}{a_0a_1} \frac{\psi_2}{\hat r} + \frac{2}{a_0a_1} \int_0^{\hat r}\frac{\psi_2}{\rho^2}\dl\rho
\\
& = -\frac{1}{a_0}\frac{\psi_1}{\hat r} + \frac{2}{a_0a_1} \frac{\psi_2}{\hat r} - \frac{2}{a_0a_1a_2}\frac{\psi_3}{\hat r},
\end{split}\label{eq:alt_int_2}
\end{align}
where we have again used Eq.~\eqref{eq:horiz_int_2} to substitute the integral of $\psi_2$ for a $\psi_3$ term. We now move our attention to the last integral. If we again use integration by parts, we find
\begin{align*}
\int_0^{\hat r} \frac{\psi_0}{\rho^3}\dl \rho & = -\frac{1}{a_0}\frac{\psi_1}{\hat r^2} - \frac{3}{a_0} \int_0^{\hat r} \frac{\psi_1}{\rho^3}\dl \rho,
\\
& = -\frac{1}{a_0}\frac{\psi_1}{\hat r^2} + \frac{3}{a_0a_1}\frac{\psi_2}{\hat r^2} + \frac{6}{a_0a_1} \int_0^{\hat r} \frac{\psi_2}{\rho^3}\dl \rho. 
\end{align*}
Finally, recognising the antiderivative of $\psi_2/\rho^2$ as $-\psi_3/(a_2\rho)$ and using Eq.~\eqref{eq:horiz_int_2} on the remaining integral yields
\begin{align}
\begin{split}
\int_0^{\hat r} \frac{\psi_0}{\rho^3}\dl \rho & = -\frac{1}{a_0}\frac{\psi_1}{\hat r^2} + \frac{3}{a_0a_1}\frac{\psi_2}{\hat r^2} - \frac{6}{a_0a_1a_2}\frac{\psi_3}{\hat r^2} 
\\
& \qquad - \frac{6}{a_0a_1a_2} \int_0^{\hat r} \frac{\psi_3}{\rho^3}\dl \rho,
\\
& = -\frac{1}{a_0}\frac{\psi_1}{\hat r^2} + \frac{3}{a_0a_1}\frac{\psi_2}{\hat r^2} - \frac{6}{a_0a_1a_2}\frac{\psi_3}{\hat r^2} 
\\
& \qquad + \frac{6}{a_0a_1a_2a_3} \frac{\psi_4}{\hat r^2}.
\end{split}\label{eq:alt_int_3}
\end{align}
Using the formula for $\psi_1$ together with Eqs.~\eqref{eq:alt_int_1}--\eqref{eq:alt_int_3}, we can isolate each $\psi_k$ and obtain Eq.~\eqref{eq:alt_sol} for $C_2=C_3=C_4=0$, giving the desired result.
\end{proof}

This property can also be proven for the case where $\psi_0 = \mathcal{O}(\hat r^2)$. However, as with Lemma~\ref{lem:nice sols}, more care is required, as $\psi_0/\hat r^3 = \mathcal{O}(1/\hat r)$, which is one of the integrands in the previous formula. Hence, we must integrate from some positive $r_0>0$ to $\hat r$ to guarantee a well-defined $\psi_4$.

We now return to constructing initial data that have compact support. The idea is to find restrictions on our initial wave profile $\psi_0$ so that each $\psi_k$ has compact support. Initially, this was not so obvious since solutions of the form \eqref{eq:horiz_integral} relate $\psi_2$, $\psi_3$, and $\psi_4$ to $\psi_0$ through nested integrals. However, after Theorem~\ref{thm:alt_sol_form}, this becomes more apparent, as each $\psi_k$ is now written in terms of separate single integrals of $\psi_0$. These ideas are summarised below. 

\begin{theorem}\label{thm:compact_sol}
Suppose that $\psi_0$ is continuously differentiable for all $\hat r \ge 0$ such that $\psi_0 = \mathcal{O}(\hat r^3)$ near $\hat r = 0$ and has compact support $[b,a]$ for some $0\le b< a$. Then $\psi_0$ can generate a continuously differentiable solution to Eq.~\eqref{eq:semi scrim eqns} such that all $\psi_k$ have compact support contained in $[0,a]$ if and only if 
\begin{equation} \label{eq:compact_sol_1}
\int_b^a \psi_0(\hat r) \dl {\hat r} = 0.
\end{equation}
Furthermore, $\psi_k(0) = 0$ if and only if 
\begin{equation} \label{eq:compact_sol_2}
\int_b^a \frac{\psi_{0}(\hat r)}{\hat r} \dl {\hat r} = 0. 
\end{equation}  
\end{theorem}

\begin{proof}
The proof is an application of Theorem~\ref{thm:alt_sol_form}. From prior results, we know that there is a continuously differentiable solution for the given $\psi_0$. Thus, all that remains is to show the properties related to the compact support of $\psi_k$. For the reverse direction of the first statement, note that we require $\psi_1$ to have compact support. Using the formula for $\psi_1$, we see that for $\hat r\ge a$, 
\[ 
\psi_1(\hr) = - \frac{a_0}{\hat r} \int_b^a \psi_0(\rho) \dl \rho. 
\]
Thus, if $\psi_1(\hat r)=0$ for $\hat r\ge a$, then Eq.~\eqref{eq:compact_sol_1} must hold. 
Conversely, if Eq.~\eqref{eq:compact_sol_1} holds, then, by the compact support of $\psi_0$, we have $\psi_1(\hat r)=0$ for $\hat r\le b$ and $\hat r\ge a$. 
It therefore remains to show that we can choose $C_2$, $C_3$, and $C_4$ such that $\psi_2$, $\psi_3$, and $\psi_4$ also have support contained in $[0,a]$. 
This is achieved by setting
\begin{align*}
C_2 & = -a_1a_0 \int_b^a \frac{\psi_0(\hat r)}{\hat r} \dl {\hat r} ,
\\
C_3 & =  \frac{a_2 a_1 a_0}{2} \int_b^a \frac{\psi_0(\hat r)}{\hat r^2} \dl {\hat r} ,
\\
C_4 & = -\frac{a_3 a_2 a_1 a_0}{6} \int_b^a \frac{\psi_0(\hat r)}{\hat r^3} \dl {\hat r} ,
\end{align*}
as follows from Eq.~\eqref{eq:alt_sol} together with the fact that $\psi_0$ has compact support in $[b,a]$. 
Moreover, this is the unique choice of $C_k$ ensuring that $\psi_k$ has support contained in $[0,a]$, as can be seen by imposing $\psi_k(\hat r)=0$ for $\hat r\ge a$.

For the second statement, if we evaluate $\psi_2$, $\psi_3$, and $\psi_4$ at $\hat r = 0$, we find
\begin{align*}
C_2 = \psi_2(0) = \frac{\psi_3(0)}{a_2} = \frac{2\psi_4 (0)}{a_3a_2}.
\end{align*}
Hence, with this and the required choice of $C_2$, the statement follows immediately. 
\end{proof}

One consequence of \cref{thm:compact_sol} is that it is not possible to set $\psi_0$ to be a simple positive bump function with compact support $[b,a]$ and have all $\psi_k$ be regular at the cylinder and have compact support. At first sight, this may seem surprising, since the idealisation of such a profile by a Dirac delta distribution forms the basis of colliding plane-wave spacetimes, such as the Khan-Penrose solution  \cite{khan1971scattering}. However, the comparison is not directly applicable, as the Khan--Penrose construction is based on planar wavefronts, whereas the present setting describes gravitational wave data on $\scrim$ with an underlying spherical geometry.

It is not immediately clear how we should use Eqs.\ \eqref{eq:compact_sol_1} and \eqref{eq:compact_sol_2} to inform choices of $\psi_0(\hat r)$, outside of checking directly they are satisfied. 
% This is important for physical accuracy, as the $\psi_k$'s represent the rescaled Weyl tensor, which helps describe the influence of the gravitational wave (even if the gravitational wave on $\scrim$ is only $\psi_0$). This is clearer when we consider a timelike observer. We know that for any timelike observer, they will experience the effect of a gravitational wave for a finite period of `time'. However, if a gravitational wave originates from $\scrim$ and is not bounded away from spatial infinity, because it has lightlike 4-velocity and $\scrip$ is a lightlike hypersurface that connects to spatial infinity, its influence will cover all of $\scrip$ and a region around it, which will also include future timelike infinity. Then, because timelike infinity is infinitely far in the future of any timelike observer, the gravitational wave described earlier will remain indefinitely in the future of the observer once it is initially influenced by the wave, thus being unrealistic.  
However, we \emph{can} further refine the choice of $\psi_0$ such that all $\psi_k$ have compact support bounded away from the cylinder. To find such a procedure, we will again use the structure of solutions given in Theorem~\ref{thm:alt_sol_form}. The resulting theorem is as follows.

\begin{theorem}\label{thm:compact_sol_2}
Suppose $\psi_0$ is continuously differentiable with compact support $[b,a]$ for $0<b<a$. Then $\psi_0$ can generate a solution to Eq.~\eqref{eq:semi scrim eqns} with all $\psi_k$ continuously differentiable with compact support contained in $[b,a]$ if and only if
\begin{align} \label{eq:compact_sol_3}
\int_b^a \psi_0(\hat r) \hat r^{1-k} \dl {\hat r} = 0,
\end{align}
for $k=1,2,3,4$. 
\end{theorem}

\begin{proof}
Note that the $k=1,2$ cases of Eq.~\eqref{eq:compact_sol_3} are given in Theorem \ref{thm:compact_sol}, where we have $C_2 = 0$ enforced. Thus, only the $k=3,4$ cases of Eq.~\eqref{eq:compact_sol_3} remain. Suppose that all $\psi_k$ have compact support contained within $[b,a]$. Then, when evaluating Eq.~\eqref{eq:alt_sol} at $0 < \hr\le b$ for $k=3,4$, we find
\begin{align*}
0 & = C_3 \hr,
\\
0 & = C_4 \hr^2 + a_3C_3\hr,
\end{align*}
noting the compact support of $\psi_0$. As $\hr \neq 0$, we can conclude that $C_3 = C_4 = 0$ is enforced. If we now evaluate both $\psi_3$ and $\psi_4$ at $\hr\ge a$, we find that
\begin{align*}
0 & = a_2a_1a_0 \hr \int_b^a  \frac{\psi_0}{\rho^2}\dl \rho,
\\
0 & = \frac{a_3a_2a_1a_0}{6}\left ( \hr^2\int_b^a  \frac{\psi_0}{\rho^3}\dl \rho -3\hr \int_b^a  \frac{\psi_0}{\rho^2}\dl \rho\right ),
\end{align*}
noting that $k=1,2$ of Eq.~\eqref{eq:compact_sol_3} hold from Theorem~\ref{thm:compact_sol} and the compact support of $\psi_0$ means the portion of the integrals from 0 to $a$ and $b$ to $\hr$ are zero. Then, as $\hr \neq 0$ and $a_k\neq 0$, it follows that both the $k=3$ and $k=4$ cases must hold if $\psi_3$ and $\psi_4$ are to have compact support contained in $[b,a]$.

For the opposite direction, suppose that Eq.~\eqref{eq:compact_sol_3} holds. Then, taking all $C_k=0$ in the solution gives that all $\psi_k$ have compact support contained in $[b,a]$ when looking at the solution structure in Eq.~\eqref{eq:alt_sol}. 
\end{proof}

\subsection{Regularity conditions on $\scrim$ in the fully-compactified gauge}
Theorems \ref{thm:compact_sol} and \ref{thm:compact_sol_2} tell us how we may find regular initial data sets in the semi-compactified representation that have compact support. However, due to the nature of this gauge, it does not allow us to choose initial data sets that extend to $i^-$, which is metrically at infinity in this gauge.
To construct initial data sets on the whole of $\scrim$, we can use the fully-compactified representation, which brings all of $\scrim$ to a finite coordinate location. In this gauge, we have results that are similar to Theorems \ref{thm:compact_sol} and \ref{thm:compact_sol_2}, but apply to all of $\scrim$.

Before examining these results, we note a couple of important features for the solutions given in Eq.~\eqref{eq:integral}. When we consider $C_k \neq 0$ for any $1\le k\le 4$, we see that the factor of $s(r)^{k-5}$ will cause $\psi_k(r)$ to be singular at $r=0$. As this correlates to past timelike infinity and we are considering gravitational wave scattering, we should not expect realistic gravitational waves to create singular behaviour in this region. Thus, we assert that $C_k = 0$ for all $1\le k\le 4$. 

Briefly examining regularity at timelike infinity, the factor of $s^{k-5}$ poses an issue. However, notice that in the integrand there is a factor of $s^{4-k}$, and we integrate from $0$ to $r$. Thus, if we only assume that $\psi_{k-1}$ is continuously differentiable near $r=0$, we can find via L'H\^opital's rule that
\begin{align*}
\lim_{r\to 0} \frac{1}{s^{5-k}} \int_0^r s^{4-k}c^{1-k}\psi_{k-1}\dl \rho & = \lim_{r\to 0}\frac{ 2s^{4-k} c^{1-k}\psi_{k-1} }{  (5-k) \pi s^{4-k} c},
\\
& = \frac{2 \psi_{k-1}(0)}{(5-k)\pi}.
\end{align*}
implying that $\displaystyle\lim_{r\to 0} \psi_{k}(r) = a_{k-1} \psi_{k-1}(0)/(5-k)$. Thus, we can take this as the definition of $\psi_k(0)$.  With this in hand, we can look at the derivative of $\psi_k$ at $r=0$. Note that all functions in the formula for $\psi_k$ are differentiable for $0<r<1$, so that we can take derivatives away from $r=0$ without problems. For the derivative at $r=0$, we find that
\begin{align*}
&\lim_{r\to 0}\frac{\psi_k(r) - \psi_k(0)}{a_{k-1} r} 
\\
& =\lim_{r \to 0}\frac{  (5-k)\pi c^{k-2} \displaystyle \int_0^r s^{4-k}c^{1-k}\psi_{k-1}\dl \rho - 2\psi_{k-1}(0)s^{5-k} }{2(5-k) r s^{5-k}} .
\end{align*}
If we now apply L'H\^opital's rule to this expression and then multiply the numerator and denominator by $s^{k-4}$, we can find that
\begin{align*}
& \lim_{r\to 0}\frac{\psi_k(r) - \psi_k(0)}{a_{k-1}r} 
\\
&\quad = \lim_{r\to 0}\frac{\pi }{ (5-k)\pi r c + 2s} \left ( \frac{ \psi_{k-1}}{c}  - \psi_{k-1}(0) c \right.
\\
& \qquad\quad +\left. \frac{1}{2} (k-2) c^{k-3} s^{k-3}\int_0^r s^{4-k}c^{1-k}\psi_{k-1}\dl \rho \right ).
\end{align*}
We can see again that we are in a similar situation to before. However, the denominator behaves like $\mathcal{O}(r)$ near $r=0$ and the integral term behaves like $\mathcal{O}(r^2)$, causing this term to terminate. Hence, we only need
 \begin{align*}
 \lim_{r\to 0}\frac{\pi ( \psi_{k-1} - \psi_{k-1}(0) c^2 )}{ (5-k) \pi r c + 2s},
 \end{align*}
noting that we have multiplied by a factor of $c$, because $c\to 1$ as $r\to 0$. Applying L'H\^opital's rule gives us
\begin{align*}
\lim_{r\to 0}\frac{ \pi (\psi_{k-1} - \psi_{k-1}(0) c^2) }{ (5-k) \pi r c + 2s} & = \lim_{r\to 0}\frac{ \psi_{k-1}' + \pi \psi_{k-1}(0) cs }{ (6-k)c - \pi(5-k) rs/2} 
\\
&= \frac{1}{6-k}\psi_{k-1}'(0),
\end{align*}
allowing us to conclude $\psi_k'(0) = a_{k-1}\psi_{k-1}'(0)/(6-k)$. Using similar arguments to those above, one can also prove that $\displaystyle\psi_k'(0) = \lim_{r\to 0} \psi_k'(r)$, which means that $\psi_k$ is continuously differentiable for $0\le r<1$. Thus, for each $\psi_k$ on $0\le r<1$, we are guaranteed that it is continuously differentiable, given that $\psi_0$ is continuously differentiable. 

Now that we have generated some baseline properties of each $\psi_k$ away from the cylinder, our attention turns back to our main concern, regularity at the cylinder. At first glance, this is particularly challenging, given that the enforced choice of integration constants requires us to integrate from timelike infinity for our iterated integrals, instead of the cylinder. This means the approaches for the semi-compactified representation will not be easily extended here. 

A remedy can be found by using the semi-compactified results as a guide to the fully-compactified consideration, particularly, the alternate formulations in Theorem \ref{thm:alt_sol_form}. If an analogous form can be found here, we can easily obtain the implications needed on $\psi_0$ to produce regular $\psi_k$ at the cylinder. Furthermore, as we are integrating from timelike infinity, which we know how to guarantee regularity at, we do not need regularity at the cylinder to perform said integration by parts. Thus, with this in mind, we have the following theorem.

\begin{theorem}\label{thm:alt_sol_form_2}
Let $\psi_0(r)$ be continuously differentiable for $0\le r\le 1$. Then the continuously differentiable solutions $\psi_k(r)$, $k=1,\dots,4$ for $0\le r<1$ given in Eq.~\eqref{eq:integral}, where all $C_k=0$, are equivalent to
% \begin{widetext}
\begin{align}
\label{eq:alt_sol_2}
\begin{split}
\psi_1(r) & = \frac{\pi a_0}{ s^3(r)\sin(\pi r)} I_1(r),
\\
\psi_2(r) & = \frac{\pi a_0 a_1}{2s^3(r)\sin(\pi r)} \Big( \sin(\pi r) I_2(r)  -2\cos(\pi r) I_1(r)\Big),
\\
\psi_3(r) & = \frac{\pi a_0a_1a_2}{8s^3(r)\sin(\pi r)} \Big( \sin^2(\pi r) I_3(r) - 2\sin(2\pi r)I_2(r)
\\
&\quad   +\,4 \cos(2\pi r) I_1(r)\Big),
\\
\psi_4(r) & = \frac{\pi a_0a_1a_2a_3}{48s^3(r)\sin(\pi r)} \Big( \sin^3(\pi r) I_4(r)\\
& \quad - 6\sin^2(\pi r)\cos(\pi r) I_3(r)
\\
&\quad +\, 4\sin(\pi r)( 2\cos(2\pi r)+1)I_2(r)\\
&\quad - 8\cos(\pi r) (2\cos(2\pi r)-1) I_1(r)\Big), 
\end{split}
\end{align}
where we have defined the integrals
\begin{align*}
 \begin{split}
  I_1(r) &= \int_0^r s^3(\rho) \psi_0(\rho) \dl \rho,\\    
  I_2(r) &= \int_0^r \frac{s^2(\rho)}{c(\rho)} \cos(\pi \rho) \psi_0(\rho) \dl \rho,\\
  I_3(r) &= \int_0^r \frac{s(\rho)}{c^2(\rho)} \psi_0(\rho)\dl \rho,\\
  I_4(r) &= \int_0^r \frac{1}{c^3(\rho)}\cos(\pi \rho) \psi_0(\rho) \dl \rho.
 \end{split} 
\end{align*}
% \end{widetext}
\end{theorem}

\begin{proof}
The first integral relation is Eq.~\eqref{eq:integral}. Thus, we begin with the first new integral $I_2$. To unravel this expression, we look to change $\psi_0$ into $\psi_1$ through integration by parts. Using Eq.~\eqref{eq:integral}, we know that an anti-derivative of $s^3\psi_0$ is $2s^4c \psi_1/(\pi a_0)$. Thus, integrating this factor, 
\begin{align*}
& I_2(r)
\equiv \int_0^r s^3 \psi_0 \frac{\cos(\pi\rho)}{sc} \dl \rho
\\
& \ =  \frac{2}{\pi a_0} s^4 c \psi_1 \frac{\cos(\pi\rho)}{sc} \biggl |_0^r - \frac{2}{\pi a_0}\int_0^r s^4 c \psi_1 \diff{}{\rho} \left ( \frac{\cos(\pi \rho)}{sc} \right )\! \dl \rho.
\end{align*}
Before differentiating the expression in the integrand, notice that $sc = \sin(\pi r/2) \cos(\pi r/2) = \sin(\pi r)/2$. Thus, the term being differentiated is $2\cot(\pi \rho)$, which has the derivative $-2\pi\csc^2(\pi \rho)$. However, further notice that $\csc(\pi \rho ) = 1/ \sin(\pi \rho) = 1/(2 s(\rho)c(\rho))$. Thus, our expression simplifies to
% \begin{align*}
$$I_2(r) = \frac{2}{\pi a_0} s^3\cos(\pi r) \psi_1 + \frac{1}{a_0}\int_0^r \frac{s^2}{c} \psi_1 \dl \rho,$$
% \end{align*}
noting that when we evaluated the first term at $\rho =0$, the factor of $s^3$ causes this to be zero. This happens similarly throughout the rest of the proof, without us explicitly mentioning it. We also note that the integral appearing in the previous expression is the same as the one occurring in the definition of $\psi_2$ in Eq.~\eqref{eq:integral}, implying
% \begin{align}
\begin{equation}\label{eq:alt_int_4}
% \begin{split}
I_2(r) 
= \frac{2}{\pi a_0} s^3\cos(\pi r) \psi_1 + \frac{2}{\pi a_0a_1} s^3 \psi_2. 
 % \end{split}
% \end{align}
\end{equation}
We now look at the next integral $I_3$. Again, using that $2s^4c \psi_1/(\pi a_0)$ is an anti-derivative of $s^3\psi_0$, we obtain
\begin{align*}
 I_3(r)
& \equiv \int_0^r \frac{1}{s^2c^2} \cdot s^3 \psi_0 \dl \rho,
\\
& = \frac{2}{\pi a_0}\frac{s^2}{c} \psi_1 - \frac{2}{\pi a_0} \int_0^r s^4 c \psi_1 \diff{}{\rho} \frac{1}{s^2c^2} \dl \rho. 
\end{align*}
Then, as $1/(s^2c^2) = 4 \csc^2 (\pi \rho)$, which has derivative $-8\pi \csc^2(\pi \rho) \cot(\pi\rho) $, we can simplify the expression to 
% \begin{align*}
\begin{equation*}
 I_3(r) = \frac{2}{\pi a_0}\frac{s^2}{c} \psi_1 +\frac{16}{a_0}\int_0^r s^4 c \csc^2(\pi \rho) \cot(\pi\rho)  \psi_1 \dl \rho.
% \end{align*}
\end{equation*}
Similarly to before, we aim to replace $\psi_1$ with $\psi_2$, which can be done when we recognise that $2s^3 \psi_2/(\pi a_0)$ is an anti-derivative of $s^2\psi_1/c$ from Eq.~\eqref{eq:integral}. Using this for integration by parts and simplifying the remaining factor in the integrand, we find that
\begin{align*}
I_3(r) &= \frac{2}{\pi a_0}\frac{s^2}{c} \psi_1 + \frac{8}{\pi a_0a_1} s^3 \cot(\pi r) \psi_2 
\\
&\qquad -\,  \frac{8}{\pi a_0a_1} \int_0^rs^3 \psi_2 \diff{}{\rho} \cot(\pi \rho) \dl \rho.
\end{align*}
Lastly, as $-\pi\csc^2(\pi \rho) = -\pi/(4s^2c^2)$, which is the derivative of $\cot(\pi\rho) $, we can simplify the expression to
\begin{align}
\label{eq:alt_int_5}
\begin{split}
I_3(r) & = \frac{2}{\pi a_0}\frac{s^2}{c} \psi_1 + \frac{8}{\pi a_0a_1} s^3 \cot(\pi r) \psi_2
\\
&\qquad +\, \frac{2}{ a_0a_1} \int_0^r\frac{s}{c^2}\psi_2 \dl \rho
\\
& = \frac{2}{\pi a_0}\frac{s^2}{c} \psi_1 + \frac{4}{\pi a_0a_1} \frac{s^2}{c} \cos(\pi r)\psi_2
\\
& \qquad + \, \frac{4}{\pi a_0a_1a_2} \frac{s^2}{c} \psi_3,
\end{split}
\end{align}
noting the use of Eq.~\eqref{eq:integral} for $k=3$ and the fact that $2\csc(\pi r) = 1/(sc)$ enabling us to write $2\cot(\pi r) = \cos(\pi r)/(sc)$.  
Finally, we consider the last integral $I_4$, which we write as
\begin{align*}
I_4(r) &= \int_0^r \frac{1+ 8s^2c^2}{c^3} \cos(\pi \rho)\psi_0 \dl \rho - 8I_2(r),
\end{align*}
to simplify derivative terms and utilising Eq.~\eqref{eq:alt_int_4} on $I_2(r)$. Focusing on the new integral, we again use the anti-derivative $2s^4c\psi_1/(\pi a_0)$ of $s^3 \psi_0$ for integration by parts. Together with $\cos(\pi\rho) = 2sc\cot(\pi \rho)$ and $1/(2sc) =\csc(\pi\rho) $, we find that our expression becomes
\begin{align*}
& \int_0^r \frac{1+ 8s^2c^2}{c^3} \cos(\pi \rho)\psi_0 \dl \rho 
\\
& \quad \equiv 8\int_0^r [\csc^2(\pi\rho) + 2 ] \cot(\pi\rho) s^3\psi_0 \dl \rho
\\
& \quad = \frac{16}{\pi a_0} s^4 c [ \csc^2(\pi r) + 2 ] \cot(\pi\rho)  \psi_1 
\\
& \qquad - \, \frac{16}{\pi a_0} \int_0^r s^4 c \psi_1 \diff{}{\rho} \left ( \left [ 2 + \csc^2(\pi\rho) \right ] \cot(\pi\rho) \right ) \dl \rho.
\end{align*}
Although $\csc^2(\pi r)$ and $[\csc^2(\pi\rho)+2]\cot(\pi\rho)$ are singular at $r=0$ and $\rho=0$, respectively, the factor $s^3$ removes this divergence. Evaluating the derivative in the integrand gives $- 3\pi \csc^4(\rho \pi)$, and our relation becomes
\begin{align*}
& \int_0^r \frac{1+ 8s^2c^2}{c^3} \cos(\pi \rho)\psi_0 \dl \rho 
\\
& \quad = \frac{8}{\pi a_0} s^3 [\csc^2(\pi r) + 2 ] \cos(\pi \rho)\psi_1
\\
& \qquad +\,  \frac{48}{a_0}\int_0^r s^4c\csc^4(\pi \rho) \psi_1 \dl \rho,
\end{align*}
where we replaced $\cot(\pi \rho)$ by $\cos(\pi \rho)/(2sc)$. We again use the anti-derivative $2s^3\psi_2/(\pi a_0 )$ of $s^2\psi_1/c$ to employ integration by parts on the next integral. Simplifying the remaining factor in the integral, we find that
\begin{align*}
& \int_0^r \frac{1+ 8s^2c^2}{c^3} \cos(\pi \rho)\psi_0 \dl \rho 
\\
& \quad = \frac{8}{\pi a_0} s^3 [\csc^2(\pi r) + 2] \cos(\pi \rho)\psi_1 + \frac{24}{\pi a_0 a_1} s^3 \csc^2(\pi\rho) \psi_2
\\
& \qquad -\, \frac{24}{\pi a_0 a_1} \int_0^r s^3 \psi_2 \diff{}{\rho} \csc^2(\pi \rho) \dl \rho. 
\end{align*}
Noting that $- 2\pi\csc^2(\pi \rho) \cot(\pi \rho)$ is the derivative of $\csc^2(\pi \rho)$ and $\csc(\pi \rho) = 1/(2sc)$, the above relation becomes
\begin{align*}
& \int_0^r \frac{1+ 8s^2c^2}{c^3} \cos(\pi \rho)\psi_0 \dl \rho 
\\
& \quad = \frac{8}{\pi a_0} s^3 [\csc^2(\pi r) + 2 ] \cos(\pi \rho)\psi_1 + \frac{6}{\pi a_0 a_1} \frac{s}{c^2} \psi_2
\\
& \qquad +\, \frac{48}{a_0 a_1} \int_0^r s^3 \csc^2(\pi \rho) \cot(\pi\rho)\psi_2 \dl \rho. 
\end{align*}
We now utilise the fact that $2s^2\psi_3/(a_2\pi c)$ is an anti-derivative of $s\psi_2/c^2$ to perform one last integration by parts. Simplifying the remaining factor then gives 
\begin{align*}
& \int_0^r \frac{1+ 8s^2c^2}{c^3} \cos(\pi \rho)\psi_0 \dl \rho 
\\
& \quad = \frac{8}{\pi a_0} s^3  [ \csc^2(\pi r) + 2  ] \cos(\pi \rho)\psi_1 + \frac{6}{\pi a_0 a_1} \frac{s}{c^2} \psi_2
\\
& \qquad +\, \frac{24}{\pi a_0 a_1a_2} \frac{s^2}{c} \cot(\pi r) \psi_3 
\\
& \qquad \quad -\, \frac{24}{\pi a_0 a_1 a_2}\int_0^r \frac{s}{c^2} \psi_3 \diff{}{\rho} \cot(\pi \rho) \dl \rho. 
\end{align*}
Noting that the derivative of $\cot(\rho \pi)$ is $- \pi \csc^2 (\rho \pi) = -  \pi/(4 s^2c^2)$, using Eq.~\eqref{eq:integral} for the $k=4$ case, and with $\cot(\pi r) = \cos(\pi r)/(2sc)$, we obtain
\begin{align*}
& \int_0^r \frac{1+ 8s^2c^2}{c^3} \cos(\pi \rho)\psi_0 \dl \rho 
\\
& \quad = \frac{8}{\pi a_0} s^3  [ \csc^2(\pi r) + 2  ] \cos(\pi \rho)\psi_1 + \frac{6}{\pi a_0 a_1} \frac{s}{c^2} \psi_2
\\
& \qquad +\, \frac{12}{\pi a_0 a_1a_2} \frac{s}{c^2} \cos(\pi r) \psi_3 +\frac{6}{ a_0 a_1 a_2}\int_0^r \frac{1}{c^3} \psi_3 \dl \rho
\\
& \quad = \frac{8}{\pi a_0} s^3  [ \csc^2(\pi r) + 2 ] \cos(\pi \rho)\psi_1 + \frac{6}{\pi a_0 a_1} \frac{s}{c^2} \psi_2
\\
& \qquad +\, \frac{12}{\pi a_0 a_1a_2} \frac{s}{c^2} \cos(\pi r) \psi_3 +\frac{12}{\pi a_0 a_1 a_2a_3}\frac{s}{c^2} \psi_4.
\end{align*}
Using this in our original integral formula, we can then find the desired relation
\begin{align}
\label{eq:alt_int_6}
\begin{split}
& I_4(r) = \frac{2}{\pi a_0} \frac{s}{c^2} \cos(\pi r)\psi_1 + \frac{2}{\pi a_0 a_1}\frac{s}{c^2}(1+ 2\cos^2(\pi r) )\psi_2
\\
& \qquad +\, \frac{12}{\pi a_0 a_1a_2} \frac{s}{c^2} \cos(\pi r) \psi_3 +\frac{12}{\pi a_0 a_1 a_2a_3}\frac{s}{c^2} \psi_4.
\end{split}
\end{align}
Finally, rearranging Eqs.~\eqref{eq:alt_int_4}--\eqref{eq:alt_int_6} to express $\psi_k$ in terms of the integrals $I_1,\dots, I_4$ and doing several cosine and sine double angle formula simplifications gives Eq.~\eqref{eq:alt_sol_2}, completing the proof. 
\end{proof}

With this in hand, we now have a direct way to compare the regularity of $\psi_k$ to the properties of $\psi_0$. It is also interesting, but not surprising, that the nature of the integrals is similar to those given in Eq.~\eqref{eq:alt_sol}, in that the new integral of $\psi_0$ introduced for $\psi_k$ mimics the integral of $\psi_{k-1}$ in the other definition of $\psi_k$. However, this time, we have what appear to be `corrective' factors of $\cos(\pi \rho)$ for the $\psi_2$ and $\psi_4$ integrals. Using this representation, we now have the following result regarding regularity at the cylinder and other properties of $\psi_k$. 

\begin{theorem}\label{thm:props_psi_k}
Suppose $\psi_0$ is continuously differentiable for $0\le r\le 1$. Then, if $\psi_0 = \mathcal{O}\left ((1-r)^n \right )$ for $n\ge 2$, it can generate a solution to Eq.~\eqref{eq:main} that is continuously differentiable for all $0\le r \le 1$ if and only if 
\begin{equation}\label{eq:props_1}
\int_0^1 \sin^3(\pi r/2) \psi_0(r) \dl r =0.
\end{equation}
Furthermore, all $\psi_k(1) = 0$ if and only if 
\begin{equation}\label{eq:props_2}
\int_0^1 \frac{\sin^2(\pi r/2)}{\cos(\pi r/2)} \cos(\pi r) \psi_0(r) \dl r = 0.
\end{equation}
Lastly, if $\psi_0$ has compact support $[a,b]$ for $0< a<b <1$, then all $\psi_k$ have compact support contained in $[a,b]$ if and only if Eqs.~\eqref{eq:props_1} and \eqref{eq:props_2} hold and
\begin{align}
\int_a^b \frac{\sin(\pi r/2)}{\cos^2(\pi r/2)} \psi_0(r) \dl r & = 0,\label{eq:props_3}
\\
\int_a^b \frac{\cos(\pi r)}{\cos^3(\pi r/2)} \psi_0(r) \dl r & = 0.\label{eq:props_4}
\end{align}
\end{theorem}

\begin{proof}
The proof of each property is heavily reliant on the result of Theorem \ref{thm:alt_sol_form_2}. For the forward direction of the proof, suppose that each $\psi_k$ generated from $\psi_0$ is continuously differentiable for all $0\le r\le 1$. If we examine $\psi_1$ at $r=1$, as it is differentiable, we are guaranteed that $\psi_1 = \mathcal{O}(1)$ near $r=1$. Thus, looking at the formula for $\psi_1$ given in \eqref{eq:alt_sol_2}, we also find that
\begin{align*}
I_1(r) = \psi_1(r)\frac{s^3 \sin(\pi r)} {\pi a_0} = \mathcal{O}(1-r),
\end{align*}
due to $s(r) =\mathcal{O}(1)$ and $\sin(\pi r) = \mathcal{O}(1-r)$. Thus, taking $r \to 1^-$ gives Eq.~\eqref{eq:props_1} as a necessary condition. 

For the sufficient condition, we look at the derivative of $I_1(r)/\sin(\pi r)$ at $r=1$. But first, we must find the limit of this quotient as $r\to 1^-$ before discussing derivatives. For this, a simple application of L'H\^opital's rule gives that
\[
\lim_{r\to 1^-} \frac{I_1(r)}{\sin(\pi r)} = \lim_{r\to 1^-} \frac{s^3 \psi_0 }{\pi \cos(\pi r)} = 0.
\]
as $\psi_0 = \mathcal{O}\left ( (1-r)^n\right )$ for $n\ge 2$ near $r=1$. Thus, for the derivative at $r=1$, we only need to take the same limit but divided by another factor of $r-1$. For this, L'H\^opital's simply yields
\begin{align*}
\lim_{r\to 1^-} \frac{I_1(r)}{(r-1)\sin(\pi r)} 
& = \lim_{r\to 1^-} \frac{s^3 \psi_0 }{\pi (r-1)\cos(\pi r) + \sin(\pi r)} ,
\\
& = 0
\end{align*}
again, because $\psi_0 = \mathcal{O}\left ( (1-r)^n\right )$ for $n\ge 2$ and the denominator behaves like $\mathcal{O}\left (1-r\right )$ near $r=1$. Thus, this term is differentiable at $r=1$. For continuously differentiable, note that the derivative of $I_1(r)/\sin(\pi r)$ for $r\neq 1$ is
\begin{align*}
\frac{s^3 \psi_0}{\sin(\pi r)} -\frac{\pi \cos(\pi r) I_1(r)}{\sin^2(\pi r)}.
\end{align*}
We know that the first term goes to zero in the limit $r\to1^-$ as $\psi_0 =\mathcal{O}\left ( (1-r)^n\right )$ for $n\ge 2$. Using L'H\^opital again, we find that the second term also approaches zero. Consequently, the limit of the derivative is equal to its value at $r=1$, making $I_1(r)/\sin(\pi r)$ continuously differentiable. 

Next, we consider the other three integral terms of relevance in \eqref{eq:alt_sol_2}. Note that the factors in front of these integrals are continuously differentiable at $r=1$. We separate the cases for $\psi_0 = \mathcal{O}\left ( (1-r)^n\right )$ into $n\ge 3$ and $n=2$ to demonstrate that the $n=2$ case still holds, but we have $\psi_4 = \mathcal{O}\left ( (1-r)^2\ln (1-r) \right )$ near $r=1$, which is undesirable. (This is the same behaviour hinted at earlier for the semi-compactified representation when $\psi_0(\hr ) = \mathcal{O}(\hr^2)$.) 

When $\psi_0 = \mathcal{O}\left ( (1-r)^n\right )$ for $n\ge 3$, we see that the integrands are continuous at $r=1$ for $I_2(r)$ and $I_3(r)$, implying that these integrals themselves are continuously differentiable. For $I_4(r)$, note that it is multiplied by a factor of $\sin^2(\pi r)=\mathcal{O}\left ((1-r)^2\right )$ near $r=1$. Thus, the derivative at $r=1$ is zero. Similarly, by the product rule, the net derivative term for $r\neq 1$ will be multiplied by a factor of $\sin (\pi r)$. Then, the entire term will have the behaviour $\mathcal{O}(1-r)$, which tends to zero as $r\to 1^-$. Hence, it is continuously differentiable. 

If $\psi_0 = \mathcal{O}\left ( (1-r)^2\right )$, the analysis becomes slightly more involved. We know that $I_2(r)$ has a continuous integrand, as it has behaviour $\mathcal{O}(1-r)$. For $I_3(r)$, notice that it has an extra factor of $\sin(\pi r)$ out front. Thus, the combined derivative at $r=1$ is
\[
\lim_{r\to 1^-} \frac{\sin(\pi r)}{r-1}I_3(r) = - \pi I_3(1). 
\]
On the other hand, the derivative for $r<1$ is
\[
\frac{s}{c^2}\sin(\pi r) \psi_0 + \pi \cos(\pi r)I_3(r),
\]
which, in the limit as $r \to 1^-$, gives exactly the value above, making the derivative continuous. For $I_4(r)$, the integrand is of order $\mathcal{O}\left (1/(1-r)\right )$, meaning the integral itself has behaviour $\mathcal{O}\left( \ln (1-r)\right )$. However, the extra factor of $\sin^2(\pi r)$ compensates for this, as can be seen in the following: We already know that it is continuous when we define it by its limit as $r \to 1^-$, which is zero. Then, for the derivative, we see that the limit considered in the definition of the derivative has behaviour $\mathcal{O}\left ((1-r)\ln (1-r) \right ) \to 0$ as $r\to 1$. Thus, the derivative is $0$ at $r=1$. For the limit of the derivative, we see that
\begin{align*}
& \frac{\sin^2(\pi r)\cos(\pi r)}{c^3} \psi_0 + 2\pi \sin(\pi r)\cos(\pi r)I_4(r) 
\\
&\qquad \qquad = \mathcal{O}\left ( (1-r) \ln (1-r) \right ) \to 0,
\end{align*}
as $r \to 1^-$, giving us the desired result. 

For the second statement, $\psi_1(1)=0$ is enforced. Thus, we are only concerned with $\psi_2$, $\psi_3$, and $\psi_4$. When we evaluate these at $r=1$ using Eq.~\eqref{eq:alt_sol_2}, we can find that
\begin{align*}
\frac{\pi a_0a_1 }{2}I_2(1) = \,\psi_2(1) = \frac{1}{a_2}\psi_3(1) = \frac{2}{a_2a_3}\psi_4(1),
\end{align*}
noting that the $\sin(2\pi r) = 2\cos(\pi r) \sin(\pi r)$ in $\psi_3$ and that all other integral terms with their corresponding factors go to zero. Hence, the second statement is an immediate consequence of this fact. 

For the last statement, because $\psi_0$ has compact support $[a,b]$, any integral considered from $0$ to $1$ of $\psi_0$ is equivalent to when it is integrated from $a$ to $b$ or $0$ to $b$, as $0<a<b<1$. For the forward direction, Eqs.~\eqref{eq:props_1} and \eqref{eq:props_2} are already enforced by the condition of regularity and $\psi_2(1) = 0$. Thus, it remains to show that Eqs.~\eqref{eq:props_3} and \eqref{eq:props_4} are enforced. For this, if we evaluate $\psi_3$ and $\psi_4$ at $r$ such that $b\le r< 1$, we find that
\begin{align*}
0 & = \frac{\pi a_0 a_1 a_2 c}{4 s^2} I_3(b),
\\
0 & = \frac{\pi a_0a_1a_2a_3c}{24s^2} ( \sin(\pi r) I_4(b) - 6 \cos(\pi r)I_3(b) ),
\end{align*}  
noting the use of $\sin(\pi r ) = 2sc$. The first equation is only satisfied when Eq.~\eqref{eq:props_3} holds as $\cos(\pi r/2) \neq 0 $ for $0<b\le r<1$. Applying this to the second of the above equations shows that Eq.~\eqref{eq:props_4} must hold as well, which concludes the proof of the forward direction. For the backward direction, we can immediately see that $\psi_k(r) = 0$ whenever $0\le r\le a$ or $b\le r\le 1$, since Eqs.~\eqref{eq:props_1}--\eqref{eq:props_4} all hold and $\psi_0$ has compact support $[a,b]$, completing the proof. 
\end{proof}

The results of Theorem \ref{thm:props_psi_k} are useful in providing us with the guidance needed to ensure certain desirable properties of each $\psi_k$ given our choice of the incoming gravitational wave $\psi_0$. 
As in the semi-compactified representation, the condition $\psi_0(r) = \mathcal{O}\left((1-r)^3\right)$ near $r=1$ ensures that $m$-times differentiability of $\psi_0$ implies $m$-times differentiability of each $\psi_k$.
The proof of this relies on a couple of key ideas. At timelike infinity ($r=0$), one should justify why considering 
\begin{align*}
&\frac{1}{s^4c}\int_0^r s^3 \psi_0\dl \rho, &&&&\frac{1}{s^3}\int_0^r \frac{s^2}{c} \cos(\pi \rho) \psi_0\dl \rho,
\\ 
&  \frac{c}{s^2}\int_0^r \frac{s}{c^2} \psi_0\dl \rho, && \text{and} && \frac{c^2}{s}\int_0^r\frac{\cos(\pi \rho)}{c^3} \psi_0\dl \rho,
\end{align*}
is equivalent to 
\[
\frac{1}{r^{n}} \int_0^{r} \rho^{n-1} f_n(\rho) \dl \rho,
\]
for a wisely chosen $f_n$, for $n = 1,\dots, 4$, which is also $m$-times continuously differentiable and related to $\psi_0$. Then, one needs to show that if $f_n$ is $m$-times continuously differentiable, then so is the previous expression. For the consideration at spacelike infinity ($r=1$), we do a similar manipulation, but this time show why it is equivalent to consider the regularity of expressions with the form
\[
\hr^{n-2}\int_0^{\hr} \rho^{1-n} g_n(\rho)\dl \rho,
\]
where $\hr = 1-r$, and $g_n$, for $n=1,\dots,4$, is a carefully selected $m$-times continuously differentiable function related to $\psi_0$ such that $g_n(\hr) = \mathcal{O}(\hr^3)$ near $\hr = 0$. Then it becomes the same justification as for the semi-compactified representation. 

A further desirable property is that each $\psi_k$ has compact support bounded away from both timelike and spacelike infinity, as realistic gravitational waves are not expected to originate in these regions and therefore should not influence them. However, this requirement is very restrictive on $\psi_0$, as it must satisfy four different integral conditions. An idea of how to construct said initial data is outlined in Sec.~\ref{sec:reg conds2} in the semi-compactified representation, but the idea remains the same for the fully-compactified representation.

\subsection{Initial data sets satisfying regularity conditions: (1) $\delta$-peaks} \label{sec:reg conds}
We now briefly discuss the constraints imposed by Theorems~\ref{thm:compact_sol},~\ref{thm:compact_sol_2}, and~\ref{thm:props_psi_k} on $\psi_0$ and the resulting initial data set in the specific case of a mode $l=2$, $m=0$. In similar problems, the simplest ingoing wave considered for intuition is a single $\delta$-distribution.
In this case, setting $\psi_0(\hr) = \delta(\hr-r_1)$ for some $r_1 \geq 0$ results in singular behaviour at $\hr=0$, as \cref{eq:compact_sol_1} is not satisfied. However, we can choose a combination of $\delta$-distributions. 
For example,
\begin{equation} \label{eq:sc double delta}
	\psi_0(\hr) = \delta(\hr-r_1) - \delta(\hr - r_2),
\end{equation}
with $r_1 > r_2 > 0$ satisfies \cref{eq:compact_sol_1}.
Although setting $\psi_0$ to be a combination of $\delta$-distributions obviously fails some of the basic requirements we have assumed of the initial data, it still provides a good way to understand the basic behaviour of the system.
If $\psi_0$ is assumed to be a continuously differentiable function, then it must have at least one positive and one negative component to potentially satisfy \cref{eq:compact_sol_1}.
Setting these two components to be $\delta$-distributions, we are able to analytically solve for the other $\psi_k$ to get a qualitative idea of the form of the initial data set. 
If we choose $\psi_0$ as in \cref{eq:sc double delta}, with $r_1 = 2/3$ and $r_2 = 1/3$, we obtain the initial data set shown in \cref{fig:sc two delta}.

We may want $\psi_k$ to vanish at $\hr=0$, in which case we need to satisfy \cref{eq:compact_sol_2}. 
This is not possible with only two $\delta$-distributions, but a linear combination of three gives enough degrees of freedom to satisfy both conditions. 
For example, take
\begin{equation} \label{eq:sc three delta}
	\psi_0(\hr) = \delta(\hr - r_1) - A \delta(\hr - r_2) + (A-1) \delta(\hr-r_3),
\end{equation} 
with $r_i > 0$, $r_i \neq r_j$, which clearly satisfies \cref{eq:compact_sol_1}. 
To satisfy \cref{eq:compact_sol_2}, we choose $A = \frac{r_2(r_1 - r_3)}{r_1(r_2 - r_3)}$.
An example of such a solution set is given in \cref{fig:sc three delta}.

We can also choose $\psi_0$ such that all $\psi_k$ have compact support away from the cylinder via \cref{thm:compact_sol_2}. Thus, we must ensure that each instance of \cref{eq:compact_sol_3} is satisfied. For this, we start with prescribing 
\begin{equation} \label{eq:sc five delta}
    \psi_0(\hr) = \sum_{i=1}^4 A_i\delta(\hr-r_i) - \delta(\hr - r_5),
\end{equation}
for any given $r_i>0$ such that $r_i \neq r_j$ for $i\neq j$. Enforcing \cref{eq:compact_sol_3}, we find that the choice of our coefficients $A_i$ is 
\begin{align}
    A_i = \frac{r_i^3}{r_5^3}\prod_{\substack{j\le 4 \\ j\neq i}} \frac{r_j - r_5}{r_j - r_i}.
\end{align}
An example of a solution of this form is given in \cref{fig:sc 5 delta}, with  $r_1=1/4$, $r_2 = 1/3$, $r_3=1/2$, $r_4=2/3$, and $r_5 = 3/4$.
\begin{figure}[htbp]
	\centering
	\begin{subfigure}{0.45\textwidth}
		\includegraphics[width=\textwidth]{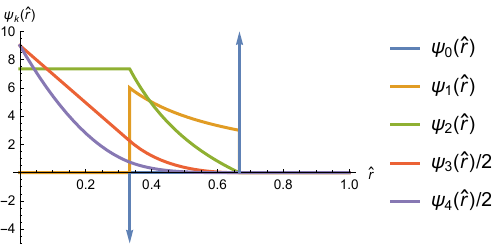}
		\caption{Resultant data set when $\psi_0(\hr) = \delta\left(\hr - \frac{2}{3}\right) - \delta\left(\hr - \frac{1}{3}\right)$.}
		\label{fig:sc two delta}
	\end{subfigure}
	\begin{subfigure}{0.45\textwidth}
		\includegraphics[width=\textwidth]{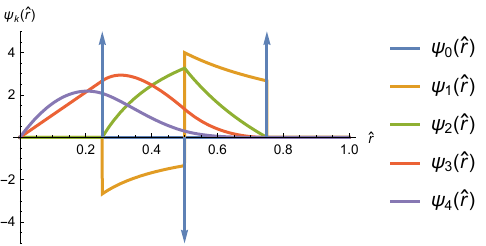}
		\caption{Resultant data set when $\psi_0(\hr) = \delta\left(\hr - \frac{3}{4}\right) - \frac{4}{3} \delta\left(\hr - \frac{1}{2}\right) + \frac{1}{3} \delta\left(\hr - \frac{1}{4}\right)$.}
		\label{fig:sc three delta}
	\end{subfigure}
    \begin{subfigure}{0.45\textwidth}
		\includegraphics[width=\textwidth]{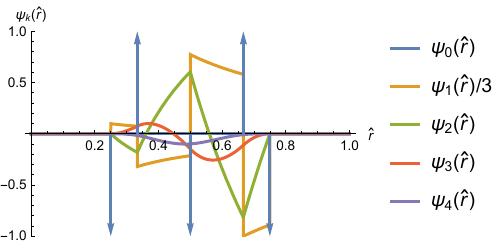}
		\caption{Resultant data set when $\psi_0(\hr) = \sum_{i=1}^4 A_i \delta(\hat r - r_i) - \delta(\hat r - r_5)$.}
		\label{fig:sc 5 delta}
	\end{subfigure}
	\caption{}
	\label{fig:sc delta sols}	
\end{figure}

With \cref{thm:props_psi_k}, we have also obtained integral conditions on the initial data for the fully-compactified gauge analogous to \cref{thm:compact_sol,thm:compact_sol_2}. 
We can get an idea of what initial data sets satisfying these conditions look like by taking $\psi_0$ as a linear combination of $\delta $-distributions.
Take
\begin{equation}
	\psi_0(r) = \delta(r-r_1) - A \delta(r-r_2),
\end{equation}   
where $A > 0$, $r_1, r_2 \in (0,1)$, and $r_1 \neq r_2$. 
To satisfy \cref{eq:props_1}, we require $A = s(r_1)^3/s(r_2)^3$. 
Choosing $r_1 = 1/3$ and $r_2 = 2/3$, and solving for the other components, we obtain the initial data set shown in \cref{fig:fc two delta}.

\begin{figure}[htbp]
	\centering
	\begin{subfigure}{0.45\textwidth}
		\includegraphics[width=\textwidth]{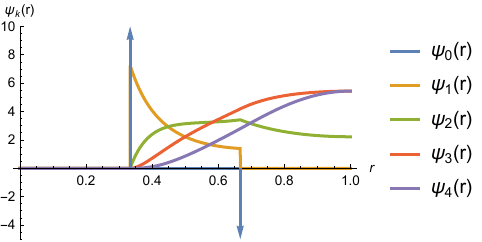}
		\caption{Resultant data set for $\psi_0$ of the form \linebreak 
        $\psi_0(r) = \delta\left(r - \frac{2}{3}\right) - A \delta\left(r - \frac{1}{3}\right)$.}
		\label{fig:fc two delta}
	\end{subfigure}
	\begin{subfigure}{0.45\textwidth}
		\includegraphics[width=\textwidth]{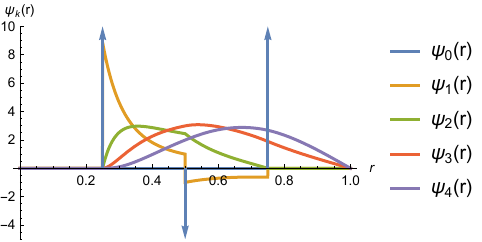}
		\caption{Resultant data set for $\psi_0$ of the form \linebreak 
        $\psi_0(r) = \delta\left(r - \frac{3}{4}\right) - A \delta\left(r - \frac{1}{2}\right) + B \delta\left(r - \frac{1}{4}\right)$.}
		\label{fig:fc three delta}
	\end{subfigure}
    \begin{subfigure}{0.45\textwidth}
		\includegraphics[width=\textwidth]{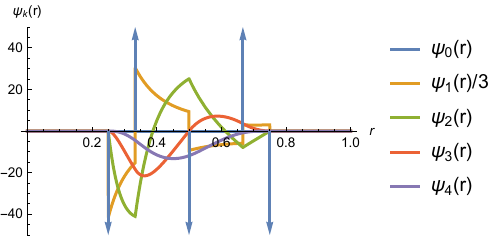}
		\caption{Resultant data set for $\psi_0$ of the form \linebreak 
        $\psi_0(r) = \sum_{i=1}^4 B_i\delta(r-r_i) - \delta(r - r_5)$.}
		\label{fig:fc five delta}
	\end{subfigure}
	\caption{}
	\label{fig:fc delta sols}	
\end{figure}

Similarly to the semi-compactified representation, we may want $\psi_k$ to vanish at $r=1$. 
Hence, we also need to satisfy \cref{eq:props_2}. 
For this, we require a linear combination of three $\delta$-distributions.  
Here, we take one of the form
\begin{equation} \label{eq:fc three delta}
	\psi_0(r) = \delta(r - r_1) - A \delta(r - r_2) + B \delta(r-r_3),
\end{equation} 
with $r_i \in (0,1)$, $r_i \neq r_j$, and $A, B>0$.
From \cref{eq:props_1,eq:props_2}, we obtain a system of equations for $A$ and $B$. 
Solving this in the case $r_1 = 1/4$, $r_2 = 1/2$, $r_3 = 3/4$ yields
\begin{equation*}
	A = 4\sqrt{2} \sin\left(\frac{\pi}{8}\right)^3, \quad B = \tan\left(\frac{\pi}{8}\right)^3.
\end{equation*}
With this choice of $\psi_0(r)$, we can solve for the other components $\psi_k$, in which case we obtain the initial data set shown in \cref{fig:fc three delta}.

Lastly, we can choose the compact support of our initial data sets to lie entirely within the interval $(0,1)$, and, hence, away from $I^-$ and $i^-$. For this, we take a linear combination of five $\delta$-distributions, like 
\begin{equation} \label{eq:fc five delta}
    \psi_0(r) = \sum_{i=1}^4 B_i\delta(r-r_i) - \delta(r - r_5),
\end{equation}
for each $r_i \neq r_j$ for $i\neq j$, impose \cref{eq:props_1,eq:props_2,eq:props_3,eq:props_4}, and solve for each $B_i$. If we choose $r_1=1/4$, $r_2 = 1/3$, $r_3=1/2$, $r_4=2/3$, and $r_5 = 3/4$, we find that the appropriate choice of each $B_i$ is  
% \begin{align*}
%     B_1 & = -\frac{1}{\sqrt{2}} \cos^6\left ( \frac{\pi}{8}\right ),
%     &
%     B_2 & = \frac{3}{4} \cos^3\left ( \frac{\pi}{8}\right ),
%     \\
%     B_3 & = -\frac{1}{2\sqrt{2}} \cos^3\left ( \frac{\pi}{8}\right ),
%     &
%     B_4 & = \frac{1}{4\sqrt{3}} \cos^3\left ( \frac{\pi}{8}\right ).
% \end{align*}
\begin{align*}
    B_1 & = -16 \sqrt 2 \cos^6\left ( \frac{\pi}{8}\right ),
    &
    B_2 & = 24 \cos^3\left ( \frac{\pi}{8}\right ),
    \\
    B_3 & = - 8 \sqrt 2 \cos^3\left ( \frac{\pi}{8}\right ),
    &
    B_4 & = \frac{8}{\sqrt 3} \cos^3\left ( \frac{\pi}{8}\right ).
\end{align*}
A plot of this solution can be seen in \cref{fig:fc five delta}.

\subsection{Initial data sets satisfying regularity conditions: (2) piecewise polynomials} \label{sec:reg conds2}

In this subsection, we outline a procedure for constructing $\psi_0$ in the semi-compactified representation such that the resulting $\psi_k$ have compact support bounded away from spacelike infinity, with the discussion again restricted to the case $l=2$, $m=0$. An identical procedure can be carried out for the fully-compactified representation, since the number of conditions is the same. We require $\psi_0$ to obey the four integral conditions in Eq.~\eqref{eq:compact_sol_3}. Randomly choosing $\psi_0$ such that Eq.~\eqref{eq:compact_sol_3} holds is highly unrealistic. Instead, we propose the following approach, which naturally extends the construction considered in the previous subsection.

First, pick five functions, say $\zeta_j$ for $1\le j \le 5$, with compact support $[b_j,a_j]$ such that $0<b_j<a_j$. Then, we write $\psi_0$ as a linear combination of these five functions. More precisely,
\begin{equation}
\psi_0 = \sum_{j=1}^4 A_j \zeta_j - \zeta_5,
\end{equation}
for arbitrary $A_j$. Imposing Eq.~\eqref{eq:compact_sol_3} results in the following four equations for the four unknowns $A_1,\dots,A_4$,
\begin{align}\label{eq:matrix_1}
\sum_{j=1}^4 A_j \zeta_{ij} = \zeta_{i5},
\end{align}
where we define 
\begin{align}\label{eq:matrix_2}
\zeta_{ij} = \int_{b_j}^{a_j} \frac{\zeta_j}{\rho^i}\dl \rho. 
\end{align}
Then, if the matrix $(\zeta_{ij})_{1\le i,j\le4}$ is invertible, the coefficients $A_j$ can be determined exactly so that Eq.~\eqref{eq:matrix_1}, and therefore Eq.~\eqref{eq:compact_sol_3}, hold.

We now illustrate a simple example where we can calculate an exact set of initial data functions. For this, let us consider the functions
\begin{align}
\zeta_j (\hr) & =  \frac{\hr^3}{B_j}  \left(  \hr - j-2\right)^2\left( \hr-j\right)^2 \label{psi_test}
\\
& = \frac{\hr^3}{B_j}\left( \hr^4 -4(j+1)\hr^3+2[3(j+1)^2-1]\hr^2 \right. \notag
\\
& \qquad\qquad \left. -\, 4j(j+1)(j+2)\hr + j^2(j+1)\right),\notag
\end{align}
for $j\le \hr\le j+2$ and zero elsewhere, where we choose $(B_1, B_2, B_3, B_4, B_5) = (5,15,30,50,100)$ to give each $\zeta_j$ approximately the same maximum value. A plot of each $\zeta_j$ is shown in \cref{example_plot}.
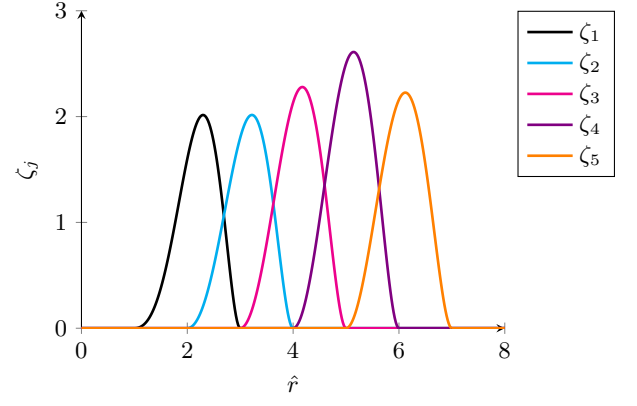
\begin{figure}
\centering
\begin{tikzpicture}
\begin{axis}
[
xtick = {0,2,4,6,8},
xlabel = {$\hr$},
ylabel = {$\zeta_j$},
ytick = {0,1,2,3},
axis lines = left,
xmin = 0,
xmax = 8,
ymin = 0,
ymax = 3,
x = 0.7cm,
y = 1.4cm,
legend pos = outer north east
]
\addplot[
samples= 2,
color=black,
domain=0:1,
line width = 1pt,
forget plot,
]
{0};
\addplot[
samples= 2,
color=black,
domain=3:8,
line width = 1pt,
forget plot,
]
{0};
\addplot[
samples= 100,
color=black,
domain=1:3,
line width = 1pt,
]
{x^3*(x-1)^2*(x-3)^2/5};
\addlegendentry{$\zeta_1$};
\addplot[
samples= 2,
color=cyan,
domain=0:2,
line width = 1pt,
forget plot,
]
{0};
\addplot[
samples= 2,
color=cyan,
domain=4:8,
line width = 1pt,
forget plot,
]
{0};
\addplot[
samples= 100,
color=cyan,
domain=2:4,
line width = 1pt,
]
{x^3*(x-2)^2*(x-4)^2/(15)};
\addlegendentry{$\zeta_2$};
\addplot[
samples= 2,
color=magenta,
domain=0:3,
line width = 1pt,
forget plot,
]
{0};
\addplot[
samples= 2,
color=magenta,
domain=5:8,
line width = 1pt,
forget plot,
]
{0};
\addplot[
samples= 100,
color=magenta,
domain=3:5,
line width = 1pt,
]
{x^3*(x-3)^2*(x-5)^2/(30)};
\addlegendentry{$\zeta_3$};
\addplot[
samples= 2,
color=violet,
domain=0:4,
line width = 1pt,
forget plot,
]
{0};
\addplot[
samples= 2,
color=violet,
domain=6:8,
line width = 1pt,
forget plot,
]
{0};
\addplot[
samples= 100,
color=violet,
domain=4:6,
line width = 1pt,
]
{x^3*(x-4)^2*(x-6)^2/(50)};
\addlegendentry{$\zeta_4$};
\addplot[
samples= 2,
color=orange,
domain=0:5,
line width = 1pt,
forget plot,
]
{0};
\addplot[
samples= 2,
color=orange,
domain=7:8,
line width = 1pt,
forget plot,
]
{0};
\addplot[
samples= 100,
color=orange,
domain=5:7,
line width = 1pt,
]
{x^3*(x-5)^2*(x-7)^2/(100)};
\addlegendentry{$\zeta_5$};
\end{axis}
\end{tikzpicture}
\caption{A plot of each $\zeta_j$ given in Eq.~\eqref{psi_test}.}\label{example_plot}
\end{figure}

Following the procedure outlined, we find that the appropriate choice of our coefficients $A_j$ is $(A_1, A_2, A_3, A_4) =(-1/20, 3/5, -9/5, 2)$. The resulting non-zero parts of $\tilde\psi_0$ in the horizontal gauge for the $l= 2$ mode are given in the Appendix, see Eq.~\eqref{eq:psi_0_example}. However, as will be mentioned later, we prefer to run numerical simulations in the diagonal gauge, since the sets where hyperbolicity is lost in the physically relevant computational domain are $I^\pm$ (unlike the horizontal gauge, where it is lost on all of $\scri^\pm$). Using the relation given in Eq.~\eqref{eq:psitrafo} and noting that $\scrim$ in one gauge corresponds to $\scrim$ in another gauge, we can calculate the non-zero part of our ingoing wave, $\psi_0$, to be Eq.~\eqref{eq:psi_0_example_2} with the corresponding $\psi_k$ in Eqs.~\eqref{eq:psi_1_example}--\eqref{eq:psi_4_example}. The accompanying plots of each $\psi_k$ are given in Fig.~\ref{psi_k_plots}.

\begin{figure}
\vspace{0.5cm}
\centering
\includegraphics[width = 0.45\textwidth]{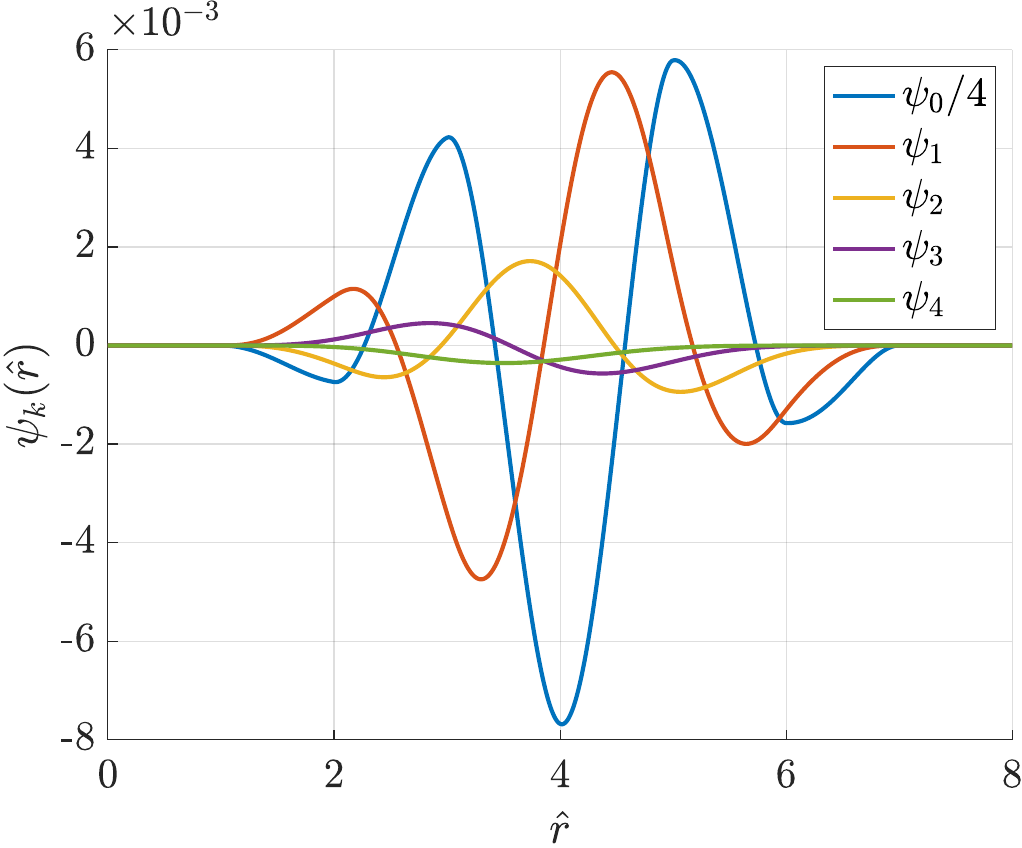}
\caption{Plot of all $\psi_k$ in the diagonal semi-compactified representation, constructed so that they have compact support bounded away from the cylinder. Note that $\psi_0$ has been scaled down so all $\psi_k$ are discernible.} \label{psi_k_plots}
\end{figure}
Obviously, there are some caveats to this method. Firstly, it is generally not possible to obtain an exact form for $\psi_0$, as precise values for all $\zeta_{ij}$ in Eq.~\eqref{eq:matrix_2} or the fully-compactified equivalent are highly unlikely. And even in the simplest cases, like the one we have given, it is also impractical and time-consuming to calculate $\psi_0$ by hand (although computer algebra programmes can manage this problem). Thus, representing more complicated behaviour generally requires these calculations to be carried out numerically. Secondly, precise control over the final form of $\psi_0$ is only achieved through prescribing some of its characteristic features, rather than specifying it exactly. Nevertheless, the method relies only on elementary linear algebra and is therefore straightforward to implement numerically.

\subsection{Numerical implementation}\label{initial_data_numeric}
In the preceding subsections, we found how to determine the initial data in integral form. 
However, in general, these integrals will not be elementary. 
Additionally, we would like to develop methods that can likely be generalised to the nonlinear case. As the principal parts of the Bianchi equation's components are naturally linear, the generalisation to the nonlinear case only introduces source terms involving connection components. With this generalisation in mind, we now turn to numerics to implement the construction of initial data.

So far, we have written the equations on $\scrim$ as a system of ODEs. This is because we introduced a basis of SWSH, which transformed the system from one of PDEs to a set of decoupled ODEs for each harmonic. This will not carry over to the nonlinear case, as modes will mix with each other. In our numerical implementation, we therefore no longer expand our functions in the SWSH basis, but instead write the equations in terms of the $\hat{\eth}'$ operator---the operator associated with the unit 2-sphere.

We use the Python package COFFEE, which provides the required numerical methods, including a numerical implementation of the $\hat{\eth}$ and $\hat{\eth}'$ operators using a spectral method based on SWSHs \cite{doulis2019, beyer2016}. We employ the $\hat{\eth}'$ operator numerically (referring the reader to \cite{beyer2016} for a discussion of how this is done) treating the equations as a system of evolution equations in the $r$-direction for functions on the spheres. Throughout, we assume axisymmetry and hence neglect the $\phi$ coordinate. Consequently, all SWSHs with $m \neq 0$ vanish.

The $r$ and $\theta$ coordinates are discretised to form a grid of equidistant points $(r_m, \theta_n)$, where $m = 1, \dots, M$ and $n = 1, \dots, N$.
We ensure that the $r$-steps are not too large relative to the size of the $\theta$ grid by fixing the factor between them, called the CFL; that is, 
\begin{equation*}
	 \frac{\Delta r}{\Delta \theta} = \text{CFL} \leq 0.2.
\end{equation*}

If we omit the step of expanding to a single harmonic before restricting the equations to $\scrim$, the equations take the form
\begin{equation} \label{eq:sc for numerics}
	\partial_{\hr} \psi_k  = \frac{k-2}{\hr} \psi_k - \frac{1}{\hr}\hat{\eth}' \psi_{k-1} = f_k(\hr, \psi_k, \psi_{k-1} )
\end{equation}
in the semi-compactified horizontal representation, and
\begin{equation} \label{eq:fc for numerics}
\begin{split}
	\partial_r \psi_k &= - \frac{\pi(3 + (7 - 2k) \cos(\pi r))}{2 \sin(\pi r)} \psi_k + \frac{\pi}{\sin(\pi r)} \hat{\eth}'\psi_{k-1}
    \\
    &= f_k(r, \psi_k, \psi_{k-1})
\end{split}
\end{equation}
in the fully-compactified representation.
On a given $r$-slice, we usually approximate the $\partial_r \psi_k$ derivatives using the classical fourth-order Runge--Kutta (RK4) method, except at singular points as described in the next section.

\subsubsection{Avoiding singular points}
At past timelike infinity $i^-$ and the bottom of the cylinder $I^-$, the GCFE become singular. 
This means the equations on $\scrim$ in the \SCG\ are singular at $\hat r = 0$, and in the \FCG, they are singular at $r=0$ and $r=1$. 
This leads to division by zero in the numerical approximation of the derivative when using RK4. The issue is resolved by not evaluating the equations at these points.

To this end, when beginning the numerical evolution at $i^-$, we use the implicit Euler method for the first $r$-step and the explicit Euler method for the final $r$-step. The explicit Euler method does not require evaluating the solution at the next step, thereby avoiding singular evaluations when the subsequent $r$-value corresponds to a singular point. Similarly, the implicit Euler method does not require evaluating the solution at the current $r$-value, so beginning with an implicit Euler step avoids evaluating the solution at $i^-$. Since these are only first-order methods, the overall rate of convergence is reduced.
Usually, performing an implicit Euler step along the $r$-coordinate to compute $(\psi_k)_{n+1} \approx \psi_k(r_{n+1}, \theta)$ would require a Newton--Raphson iteration or similar method to solve the implicit equation
\begin{equation*}
    (\psi_k)_{n+1} = (\psi_k)_n + \Delta r\, f_k(r_{n+1}, (\psi_k)_{n+1}, (\psi_{k-1})_{n+1}),
\end{equation*}
where $f_k$ denotes the right-hand side of the differential equation for $\psi_k$.
Determining $(\psi_k)_{n+1}$ from this equation would generally require an iterative procedure. However, in the \FCG, only a single implicit Euler step from $r_0 = 0$ is necessary, and we assume that $\psi_0(0,\theta)=0$, i.e.\ we choose a $\psi_0$ that vanishes at $i^-$. In this case, the implicit equation can be solved directly for $\psi_k(r_1,\theta)$, yielding
\begin{equation}
	\psi_k(r_1, \theta ) \approx \frac{\pi r_1 \eth' \psi_{k-1}(r_1, \theta)}{\sin(\pi r_1) + \frac{\pi}{2} r_1 \bigl(3 + (7-2k)\cos(\pi r_1)\bigr)}.
\end{equation}

\subsubsection{Convergence tests}
In numerical analysis, a common way to test the correctness of a numerical method is through convergence tests. Exact solutions are particularly useful for this purpose. Fortunately, such solutions are straightforward to construct: We may simply prescribe $\psi_4$ and then use differentiation together with the system of ODEs to obtain the remaining $\psi_k$, including $\psi_0$.

We begin with the semi-compactified representation. For the exact solution, we choose
\begin{equation}
	\psi_4(\hr) = \begin{cases}
		\hr^2 (\hr-1)^6 & \text{if $0 < \hr < 1$,}
		\\
		0 & \text{otherwise,}
	\end{cases}
\end{equation}
which is a simple bump function, with appropriate fall-off at each end so that the other $\psi_k$ will all be regular.
This yields the initial data set
\begin{equation}
	\begin{split}
		\psi_{0}(\hr) &= 3 (\hr-1)^2 \hr^3 \left(42 \hr^3-56 \hr^2+21 \hr-2\right),
		\\
		\psi_{1}(\hr) &= -(\hr-1)^3 \hr^3 \left(28 \hr^2-21 \hr+3\right),
		\\
		\psi_{2}(\hr) &= \sqrt{\frac{3}{2}} (\hr-1)^4 \hr^3 (7 \hr-2),
		\\
		\psi_{3}(\hr) &= -3 (\hr-1)^5 \hr^3,
		\\
		\psi_{4}(\hr) &= (\hr-1)^6 \hr^2,
	\end{split}
\end{equation}
on $(0, 1)$, and $0$ elsewhere. These are the solutions for the $l=2,m=0$ mode. We now change our notation so that $\psi_k(\hr, \theta) = \psi_k(\hr) \cdot \swsh{2\!}{20}(\theta)$, in order to remove ourselves from the SWSH basis.

The system was solved `backwards' relative to the $\hr$ coordinate, beginning at $\hr = 1$ and evolving to $\hr = 0$.
As we are working in coordinates that have been inverted, $\hr = 0$ corresponds to $I^-$, which is in the future of every point on $\scri^-$, so taking steps in the negative direction with respect to the $\hr$ coordinate is still evolving towards the future in this gauge.
This has the additional advantage of avoiding the need to switch to a first-order method until the end of the evolution. Moreover, all the solutions we work with here vanish at $\hr=1$, but are not necessarily zero at the cylinder, so it is more natural to start from the end of the interval where the value of $\psi_k$ is known.
Note that restricting attention to solutions supported on $[0,1)$ rather than $[0,a)$ for some $a>0$ is arbitrary, since the underlying system of ODEs is invariant under the rescaling $\hr \mapsto a\,\hr$.

As expected from linear theory, our numerical experiments show that, for each component $\psi_k$, the $l=2$ mode is the only mode that is excited. Owing to numerical error, some other modes become non-zero, but all remain small---around $10^{-12}$.
Because of this, to visualise the solution and compute the error and convergence rate, we simply restricted to the $l=2$ mode of the computed solution.
A plot of this solution is shown in \cref{fig:sc exact}.
To test convergence to the exact solution, the resolution was successively increased by repeatedly halving the CFL parameter.
While it would be more usual to decrease the resolution in both directions, the angular direction is treated using a spectral method, and $N=32$ is already sufficiently accurate. Moreover, the same convergence rate was observed when the angular resolution was also increased.
The resulting error for $\psi_4$ is shown in \cref{fig:sc error}.

\begin{figure}[htbp]
	\centering
	\begin{subfigure}[t]{0.44\textwidth}
		\centering
		\includegraphics[scale=0.55]{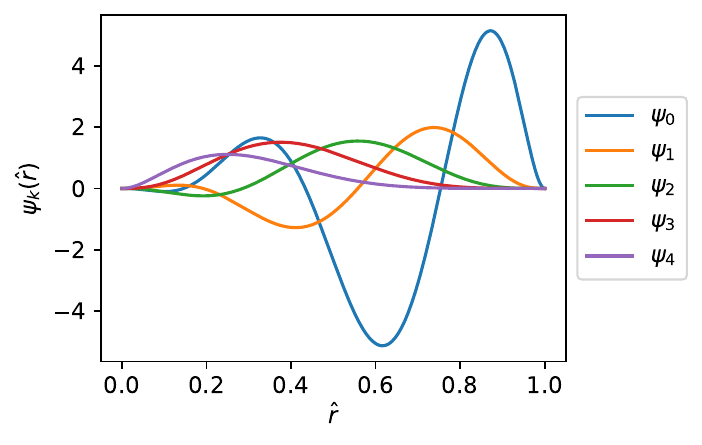}
		\caption{The $l=2$ mode of the exact solution.}
		\label{fig:sc exact}
	\end{subfigure}
	\hfill
	\begin{subfigure}[t]{0.54\textwidth}
		\centering
		\includegraphics[scale=0.55]{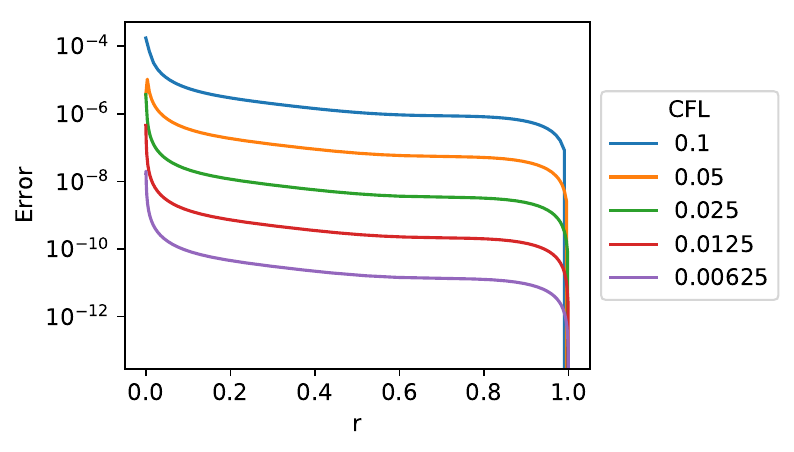}
		\caption{Error in $\psi_4$ as the resolution is increased.}
		\label{fig:sc error}
	\end{subfigure}
	\caption{Exact solution and convergence plot for the semi-compactified gauge.}
\end{figure}

For a single increase in resolution, the convergence rate was computed as 
\begin{equation}
	\log_2\!\left(\frac{\|E_N\|_2}{N}\right)
   -\log_2\!\left(\frac{\|E_{2N}\|_2}{2N}\right),
\end{equation}
where $\|\cdot\|_2$ is the $\ltwo$-norm and $E_N$ is the error when using $N$ points.
The results are shown in \cref{tab:sc convergence}. 
\begin{table}[htbp]
	\centering
	\caption{Error and convergence rate to the exact solution in the semi-compactified gauge.}
	\label{tab:sc convergence}
	$\begin{NiceArray}{lcccc}
		\hline
		\Block{2-1}{\text{CFL}} & \Block{1-2}{\psi_1} & &  \Block{1-2}{\psi_4} &
		\\
		 & \text{Error} & \text{Convergence}& \text{Error} & \text{Convergence}
		\\
		\hline
		0.1		& -17.66 & 		& -13.80 & 		\\
		0.05	& -21.92 & 4.26 & -16.67 & 2.87 \\
		0.025	& -25.09 & 3.17 & -18.93 & 2.26 \\
		0.0125	& -29.20 & 4.10 & -22.26 & 3.33 \\
		0.00625	& -33.43 & 4.23 & -25.32 & 3.06 \\
		\hline
	\end{NiceArray}$
\end{table}
Notably, the convergence rate is not very consistent as the resolution is increased.
This is due to the region near $\hr = 0$, where the $1/\hr$ factor in \cref{eq:sc for numerics} becomes very large, leading to a rapid accumulation of numerical errors. This effect is clearly visible in \cref{fig:sc error}. If points near $\hr = 0$ are excluded, the convergence rate improves significantly. This is illustrated in \cref{tab:sc convergence 2}, where points in the interval $[0,0.1)$ are omitted.

\begin{table}[htbp]
	\centering
	\caption{Error and convergence rate to the exact solution on the interval $(0.1, 1)$ in the semi-compactified gauge.}
	\label{tab:sc convergence 2}
	$\begin{NiceArray}{lcccc}
		\hline
		\Block{2-1}{\text{CFL}} & \Block{1-2}{\psi_1} & &  \Block{1-2}{\psi_4} &
		\\
		 & \text{Error} & \text{Convergence}& \text{Error} & \text{Convergence}
		\\
		\hline
		0.1		& -19.40 & 		& -19.81 & 		\\
		0.05	& -23.38 & 3.98 & -23.80 & 3.98 \\
		0.025	& -27.37 & 3.99 & -27.79 & 3.99 \\
		0.0125	& -31.38 & 4.01 & -31.80 & 3.99 \\
		0.00625	& -35.38 & 4.00 & -35.80 & 3.99 \\
		\hline
	\end{NiceArray}$
\end{table}

Similarly to our discussion of the semi-compactified gauge, we next test convergence of our methods in the fully-compactified representation. Here, we choose an exact solution of the form $\psi_k(r, \theta) = \psi_k(r)\, \swsh{2-k}{20}(\theta)$, setting $\psi_0(r, \theta) = \psi_0(r)\, \swsh{2}{20}(\theta)$, and then test that the other components $\psi_k$ converge at the correct rate.
Once again, the easiest way to generate a non-trivial exact data set is to prescribe $\psi_4$.   
Choosing $\psi_4(r) = \sin^2(\pi r)$, yields the following set of initial data:
\begin{equation}
	\begin{split}
		\psi_0(r) &= \frac{15}{64} \sin ^2\left(\frac{\pi  r}{2}\right) \cos ^4\left(\frac{\pi  r}{2}\right)
        \\
        & \quad (115 \cos (\pi  r)+38 \cos (2 \pi  r)+77 \cos (3 \pi  r)+26),
		\\
		\psi_1(r) &= \frac{1}{8} \sin ^2\left(\frac{\pi  r}{2}\right) \cos ^4\left(\frac{\pi  r}{2}\right)
        \\
        &\quad (32 \cos (\pi  r)+77 \cos (2 \pi  r)+51),
		\\
		\psi_2(r) &= \sqrt{\frac{3}{2}} \sin^2\left(\frac{\pi  r}{2}\right) \cos ^4\left(\frac{\pi  r}{2}\right) (7 \cos (\pi  r)+1),
		\\
		\psi_3(r) &= \frac{3}{4} \sin ^2(\pi  r) (\cos (\pi  r)+1),
		\\
		\psi_4(r) &= \sin ^2(\pi  r).
	\end{split}
\end{equation}
We then set $\psi_0(r, \theta) = \psi_0(r)\ \swsh{2\!}{20}(\theta)$ for the above choice of $\psi_0$, and found the other $\psi_k$ numerically, evolving from $r = 0$ to $r = 1$.
As in the semi-compactified representation, for each component $\psi_k$, only the $l=2$ mode became significantly excited, while some other modes became non-zero, but stayed below about $10^{-12}$.
Consequently, our plots and error calculations are again restricted to the $l = 2$ mode. 

\cref{fig:test soln} shows the $l=2$ modes for a CFL of $0.2$ and an angular resolution of $N = 32$. Comparing with the exact solution given above, we compute how the error in the numerical approximation of $\psi_4$ changes as the resolution in the $r$-direction is successively doubled. This results in the convergence plot shown in \cref{fig:fc convergence}.
The convergence rate turns out to be between second and third order; see \cref{tab:fc convergence}.
\begin{figure}[htbp]
	\centering
	\begin{subfigure}[t]{0.44\textwidth}
		\centering
		\includegraphics[scale=0.55]{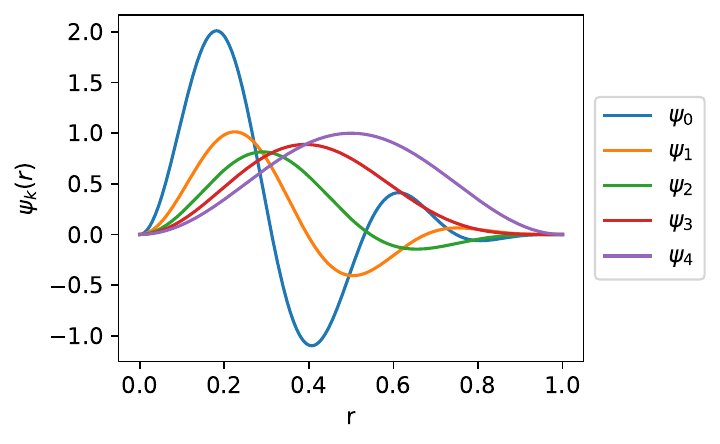}
		\caption{The $l=2$ mode for an exact solution in the fully-compactified gauge, found numerically with a CFL of $0.2$.}
		\label{fig:test soln}
	\end{subfigure}
	\begin{subfigure}[t]{0.54\textwidth}
		\centering
		\includegraphics[scale=0.55]{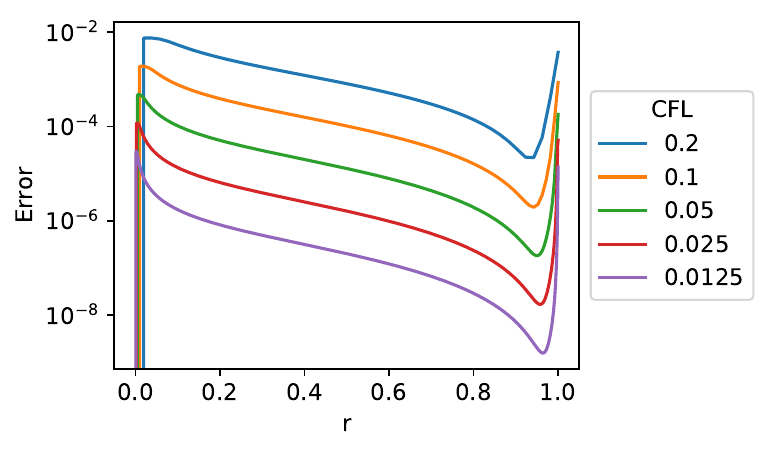}
		\caption{Absolute error in $\psi_4$ computed for different CFL values.}
		\label{fig:fc convergence}
	\end{subfigure}
	\caption{Exact solution and convergence plot for the fully-compactified gauge.}
\end{figure}
\begin{table}[htbp]
	\centering
	\caption{Error and convergence rate to the exact solution in the fully-compactified gauge.}
	\label{tab:fc convergence}
	$\begin{NiceArray}{lcccc}
		\hline
		\Block{2-1}{\text{CFL}} & \Block{1-2}{\psi_1} & &  \Block{1-2}{\psi_4} &
		\\
		 & \text{Error} & \text{Convergence}& \text{Error} & \text{Convergence}
		\\
		\hline
		0.2   	& -11.66 & 		& -9.21  & 		\\
		0.1   	& -14.51 & 2.85 & -11.84 & 2.63 \\
		0.05  	& -17.45 & 2.94 & -14.55 & 2.70 \\
		0.025 	& -20.42 & 2.97 & -17.28 & 2.73 \\
		0.0125	& -23.40 & 2.98 & -20.08 & 2.80 \\
		\hline
	\end{NiceArray}$
\end{table}
This is likely due to the implicit Euler step at $r = 0$. 
Using the same exact solution, but beginning at $r = 0.5$, we can avoid the implicit Euler step.
In this case, we obtain much closer to fourth-order convergence, except near $r = 1$; this can be seen in \cref{fig:fc convergence 2} and \cref{tab:fc convergence 2}.
\begin{figure}[htb]
	\centering
	\includegraphics[scale=0.55]{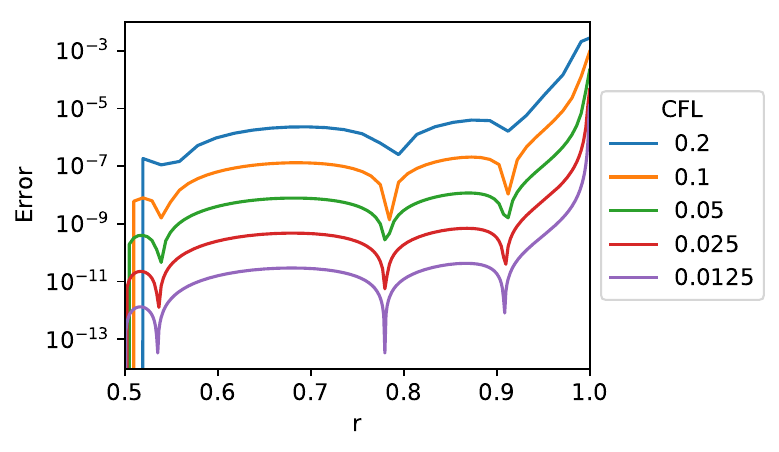}
	\caption{Error in $\psi_4$ when evolving an exact solution in the fully-compactified gauge beginning at $r = 0.5$.}
	\label{fig:fc convergence 2}
\end{figure}

\begin{table}[htb]
	\centering
	\caption{Error and convergence rate to the exact solution in the fully-compactified gauge when beginning evolution from $r = 0.5$ and ignoring points after $r = 0.9$.}
	\label{tab:fc convergence 2}
	$\begin{NiceArray}{lcccc}
		\hline
		\Block{2-1}{\text{CFL}} & \Block{1-2}{\psi_1} & &  \Block{1-2}{\psi_4} &
		\\
		 & \text{Error} & \text{Convergence}& \text{Error} & \text{Convergence}
		\\
		\hline
		0.2   	& -20.31 & 		& -18.38 & 		\\
		0.1   	& -24.40 & 4.09 & -23.06 & 4.68 \\
		0.05  	& -28.37 & 3.97 & -27.00 & 3.93 \\
		0.025 	& -32.35 & 3.98 & -30.97 & 3.97 \\
		0.0125	& -36.34 & 3.99 & -34.95 & 3.99 \\
		\hline
	\end{NiceArray}$
\end{table}

% CFL: 0.2        psi1 error: -20.31
% CFL: 0.1        psi1 error: -24.40 convergence = 4.09
% CFL: 0.05       psi1 error: -28.37 convergence = 3.97
% CFL: 0.025      psi1 error: -32.35 convergence = 3.98
% CFL: 0.0125     psi1 error: -36.34 convergence = 3.99

% CFL: 0.2        psi4 error: -18.38
% CFL: 0.1        psi4 error: -23.06 convergence = 4.68
% CFL: 0.05       psi4 error: -27.00 convergence = 3.93
% CFL: 0.025      psi4 error: -30.97 convergence = 3.97
% CFL: 0.0125     psi4 error: -34.95 convergence = 3.99

\subsubsection{Numerically generating initial data sets}
Of course, there is little point in a numerical implementation if it does not allow us to generate initial data which could not be found analytically.
We now turn our attention to the generation of initial data sets through choosing the ingoing radiation field $\psi_0$ on $\scrim$, and then solving numerically for the other components $\psi_k$. 

We begin with the semi-compactified representation. 
\cref{thm:compact_sol} implies that we cannot use a simple bump function for $\psi_0$; however, a combination of bump functions may work.
Inspired by the discussion in Sec.~\ref{sec:reg conds} about solutions involving combinations of Dirac delta distributions such as \cref{eq:sc double delta}, we may choose $\psi_0$ to be a combination of bump functions. 

We use the bump function
\begin{equation}
	b(\hr) = 
	\begin{cases}
		\sin \! \left(\frac{\pi \hr}{\Delta \hr}\right)^8 & \text{if $0 \leq \hr \leq \Delta \hr$},
		\\
		0 & \text{otherwise.}
	\end{cases},
\end{equation}
with $\Delta \hr = 1/8$. 
Note that this function has compact support $(0, \Delta \hr)$, and is $8$ times differentiable. 
A regular initial data set analogous to the choice in \cref{eq:sc double delta} for the semi-compactified representation may be generated by simply replacing the $\delta$-distributions with bump functions; for example, 
\begin{equation} \label{eq:sc two bump}
	\psi_0(\hr) = b \! \left(\hr - \frac 13\right) - b \! \left(\hr- \frac 23\right).
\end{equation}
% Finally, recalling that $\psi_4$ will contain an $r^2$ term but no $r^3$ term or higher in its expansion at $r=0$, we will choose an $l=3$ mode rather than $l=2$ to ensure that \cref{lem:reg} is satisfied.      
If we set $\psi_0(\hr, \theta) = \psi_0(\hr)\ \swsh{2\!}{20}(\theta)$ on $\scrim$, and evolve numerically with an angular resolution of $N = 32$ and a CFL of $0.05$, which yields a spatial resolution of $M = 204$, we obtain the initial data set shown in \cref{fig:sc two bump}.

If we choose $\psi_0$ that satisfies \cref{eq:compact_sol_2} as well, then we can obtain an initial data set that vanishes at the cylinder. 
Analogously to \cref{eq:sc three delta}, we may choose
\begin{equation} \label{eq:sc three bump}
	\psi_0(\hr) = b \!\left(\hr-\frac{1}{4}\right) - A \, b \! \left(\hr-\frac{1}{2}\right) + (A-1) \, b \! \left(\hr - \frac{3}{4}\right),
\end{equation}
which clearly satisfies \cref{eq:compact_sol_1}; that is, its integral is zero over the interval $(0,1)$. 
We then use \cref{eq:compact_sol_2} to find the value of $A$ which yields initial data that vanishes at the cylinder.  
It turns out that we need to choose $A \approx 1.3837878553899127$.
Using \cref{eq:sc three bump} with this value of $A$ and the same resolution as before yields the initial data set shown in \cref{fig:sc three bump}.

Finally, we may obtain an initial data set that has compact support by choosing $\psi_0(\hr)$ to be a combination of bumps that satisfies \cref{eq:compact_sol_3}. 
Analogously to \cref{eq:sc five delta}, we set 
\begin{equation} \label{eq:sc five bump}
    \psi_0(\hr) = \sum_{i=1}^4 B_i b(\hr - r_i) - b(\hr-r_5),
\end{equation}
where $r_1 = 1/4, r_2 = 1/3, r_3 = 1/2, r_4 = 2/3, r_5 = 3/4$, and the coefficients $B_i$ are chosen by solving the linear system imposed by \cref{eq:compact_sol_3}.
Using \cref{eq:sc five bump} with these values of $B_i$ and the same resolution as before yields the initial data set shown in \cref{fig:sc five bump}.
\begin{figure}[htbp]
	\centering
	\begin{subfigure}{0.49\textwidth}
		\includegraphics[scale=0.55]{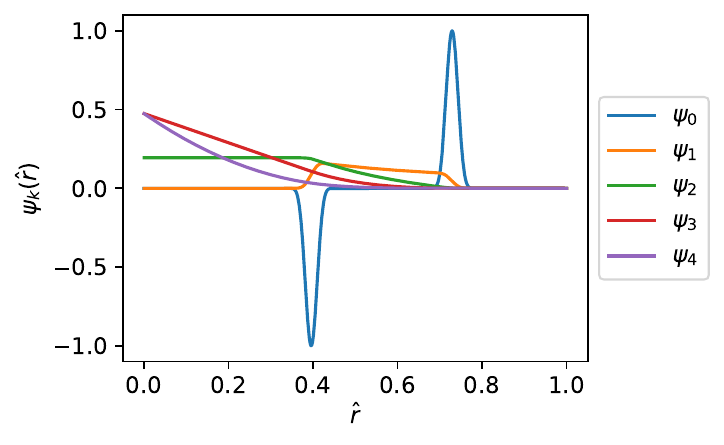}
		\caption{Initial data from choosing $\psi_0$ as in \cref{eq:sc two bump}.}
		\label{fig:sc two bump}
	\end{subfigure}
	\begin{subfigure}{0.49\textwidth}
		\includegraphics[scale=0.55]{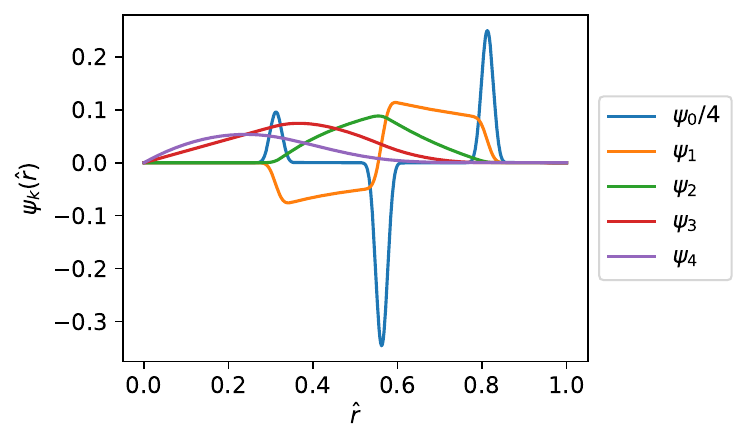}
		\caption{Initial data from choosing $\psi_0$ as in \cref{eq:sc three bump}.}
		\label{fig:sc three bump}
	\end{subfigure}
    \begin{subfigure}{0.49\textwidth}
		\includegraphics[scale=0.55]{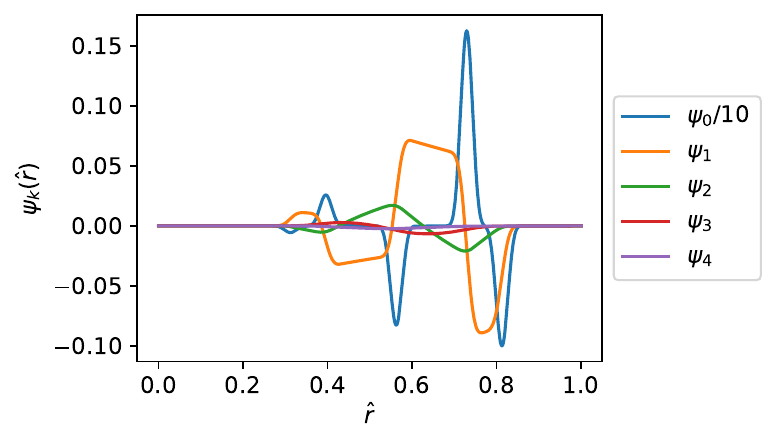}
		\caption{Initial data from choosing $\psi_0$ as in \cref{eq:sc five bump}.}
		\label{fig:sc five bump}
	\end{subfigure}
	\caption{Some numerically calculated initial data sets in the semi-compactified gauge.}
\end{figure}
It is interesting to note the similarity of \cref{fig:sc two bump,fig:sc three bump,fig:sc five bump} to \cref{fig:sc two delta,fig:sc three delta,fig:sc 5 delta}.

In the fully-compactified gauge, we can obtain full initial data sets through a similar method.
For example, a regular initial data set analogous to the choice of $\psi_0$ in \cref{eq:sc three delta} for the semi-compactified representation can be generated by choosing
\begin{equation} \label{eq:fc two bump}
	\psi_0(r) = b \! \left(r-\frac{1}{3}\right) - A \, b \! \left(r-\frac{2}{3}\right),
\end{equation}
where $A$ is fixed by the requirement that \cref{eq:props_1} holds.
We compute $A \approx 0.2620939449674576$.
If we set $\psi(r, \theta) = \psi_0(r)\ \swsh{2\!}{30}$ for this choice of $\psi_0$, with an angular resolution of $N = 32$ and a spatial resolution of $M = 204$, this yields the initial data set shown in \cref{fig:fc two bump}.

To obtain an initial data set that vanishes on the cylinder, we choose
\begin{equation} \label{eq:fc three bump}
	\psi_0(r) = b \!\left(r-\frac{1}{4}\right) - A \, b \! \left(r-\frac{1}{2}\right) + B \, b \! \left(r - \frac{3}{4}\right),
\end{equation}
and then use \cref{eq:props_1,eq:props_2} from \cref{thm:props_psi_k} to find the appropriate choices of $A$ and $B$.
Doing this, we find that $A \approx 0.37822387453047346$ and $B \approx 0.07945013784832826$.
Setting $\psi_0(r, \theta) = \psi_0(r)\ \swsh{2\!}{30}$ with this choice of $\psi _0$ and the same resolution as before yields the initial data set shown in \cref{fig:fc three bump}.

To obtain an initial data set with compact support, we may choose $\psi_0( r)$ to be a combination of bumps that satisfies \crefrange{eq:props_1}{eq:props_4}.
Analogously to \cref{eq:fc five delta}, we set 
\begin{equation} \label{eq:fc five bump}
    \psi_0(r) = \sum_{i=1}^4 B_i b(r - r_i) - b(r-r_5),
\end{equation}
where $r_1 = 1/4, r_2 = 1/3, r_3 = 1/2, r_4 = 2/3, r_5 = 3/4$, and the coefficients $B_i$ are chosen by solving the linear system imposed by \crefrange{eq:props_1}{eq:props_4}.
Using this choice of $\psi_0$ with the same resolution as before yields the initial data set shown in \cref{fig:fc five bump}.

As with the semi-compactified representation, we may note the similarity to the initial data sets seen in \cref{fig:fc two delta,fig:fc three delta,fig:fc five delta}.
\begin{figure}[htbp]
	\centering
	\begin{subfigure}{0.49\textwidth}
		\includegraphics[scale=0.55]{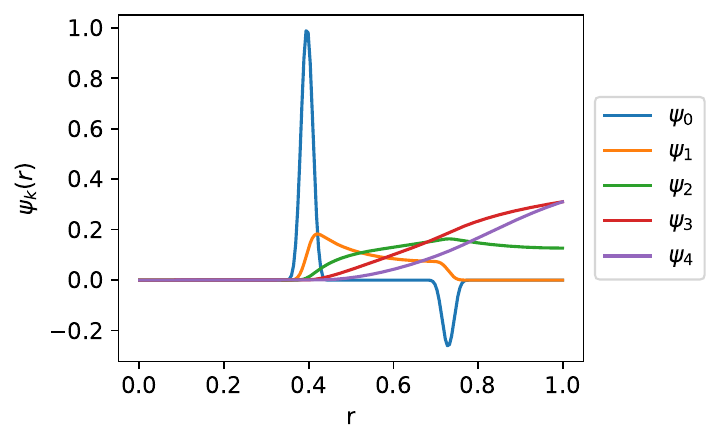}
		\caption{Initial data from choosing $\psi_0$ as in \cref{eq:fc two bump}.}
		\label{fig:fc two bump}
	\end{subfigure}
	\begin{subfigure}{0.49\textwidth}
		\includegraphics[scale=0.55]{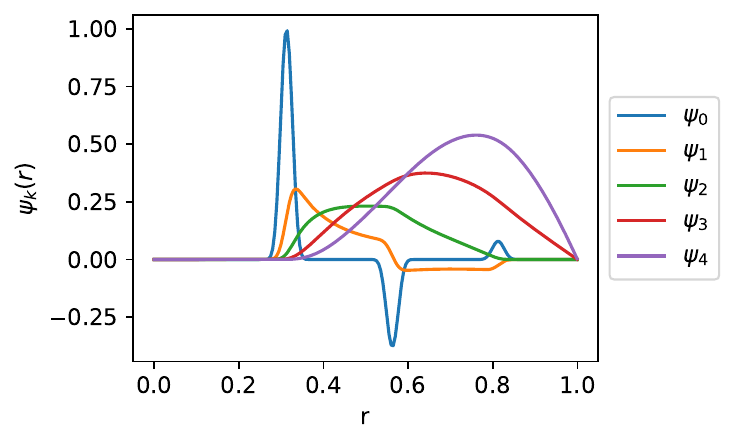}
		\caption{Initial data from choosing $\psi_0$ as in \cref{eq:fc three bump}.}
		\label{fig:fc three bump}
	\end{subfigure}
    \begin{subfigure}{0.49\textwidth}
		\includegraphics[scale=0.55]{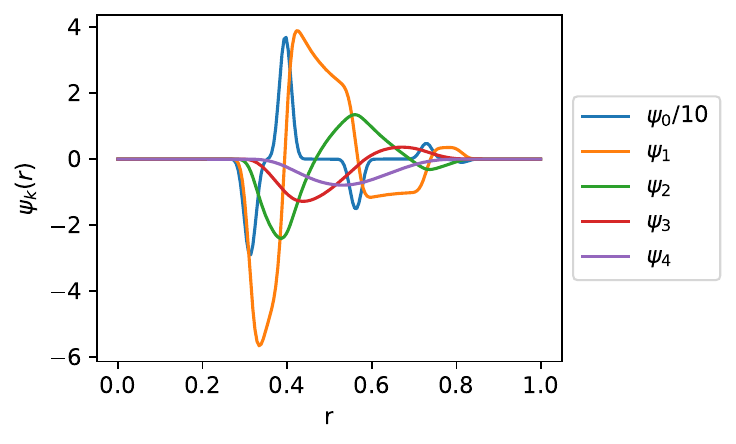}
		\caption{Initial data from choosing $\psi_0$ as in \cref{eq:fc five bump}.}
		\label{fig:fc five bump}
	\end{subfigure}
	\caption{Some numerically calculated initial data sets in the fully-compactified gauge.}
\end{figure}

Due to the singular nature of the equations at the critical points at $i^-$ and $I^-$, one might expect that for the initial data sets with compact support avoiding these points we might observe better numerical performance.
However, we did not find a difference in the number of gridpoints required or the convergence between the different choices of $\psi_0$ tested. Hence, we can conclude that our numerical method is robust enough to handle singular points in the initial data without significant loss of accuracy.

% \subsubsection{Initial data sets for global evolutions}
% The following describes the initial data sets used in the remainder of the paper, as initial data for a fully global evolution.
% These are all in the \SCG.
% We begin with initial data sets which are `trivial' in the sense that they are generated by choosing $\psi_4$ and differentiating, rather than picking $\psi_0$ and integrating. 
% These have the advantage of being nicely behaved, while still being nontrivial to evolve through the spacetime. 
% We then create a nontrivial initial data set by choosing bmup functions as described in the previous section.

\section{Evolving through the cylinder}\label{sec:IV}
Having generated initial data on $\mathscr{I}^-$, we now turn to evolving them through the cylinder and onto $\mathscr{I}^+$ to obtain the full solution in a manner akin to \cite{doulis2013second}. This presents several challenges.

In deriving our methods, we work with the \emph{semi-compactified representation}, as this avoids issues at both the physical origin and timelike infinity. We also expand the $\psi_k$ in the SWSH basis and restrict attention to the $l=2$, $m=0$ mode. In this case, the evolution equations can be written as
\begin{equation}
    \mathbf{A}^0\partial_t\mathbf{\psi} - \mathbf{A}^1\partial_r\mathbf{\psi} = \mathbf{B}\mathbf{\psi},
\end{equation}
where $\mathbf{\psi}$ is the vector of components of the spin-2 field, and $\mathbf{A}^0$, $\mathbf{A}^1$, and $\mathbf{B}$ are $5 \times 5$ matrices, with
\begin{equation}
    \mathbf{A}^0 = \begin{pmatrix}
        1+t\kappa' &0&0&0&0\\
        0&1&0&0&0\\
        0&0&1&0&0\\
        0&0&0&1&0\\
        0&0&0&0&1 - t\kappa'
    \end{pmatrix}.
\end{equation}
As all $A^\mu$ are Hermitian, the evolution equations will form a symmetric hyperbolic system when the matrix $\mathbf{A}^0$ is positive definite, that is, for $|t| < 1/\kappa'$. Outside of this region, the equations lose hyperbolicity, which is disastrous for our numerical methods. The region where hyperbolicity is lost includes the top and bottom of the cylinder at infinity and extends far into the region beyond $\mathscr{I}$, as shown in \cref{LossofHyperbolicity}.
\begin{figure}[htbp]
\centering

\begin{subfigure}{0.2\textwidth}
\centering
\begin{tikzpicture}[baseline,remember picture]

% Light red region outside (above) the parabola
\fill[red!20]
  plot[domain=2:4] (\x,{(4-\x)^2})  % along parabola
  -- (4,4)                          % up right side
  -- (2,4)                          % across top
  -- cycle;                         % back to start

% Original diagram
\draw (3,1)  -- node[anchor = west]{$\mathscr{I}^+$} (0,4);
\draw (3,0) -- (3,1);
\draw [dashed] (0,0) -- (3,0);
\filldraw (3,1) circle (1pt) node[anchor = west]{$I^+$};

\draw[red] plot[smooth,domain=2:4] (\x,{(4-\x)^2});

\end{tikzpicture}
\end{subfigure}
\hfill
\begin{subfigure}{0.2\textwidth}
\centering
\begin{tikzpicture}[baseline,remember picture]

% Light red filled region above the line
\fill[red!20] (0,2) -- (3,2) -- (3,3.5) -- (0,3.5) -- cycle;

% Red boundary line
\draw[red] (0,2) -- (3,2);

\node (n1) at (1.7,2.2) {$\mathscr{I}^+$};

\draw (3,0) -- (3,2);
\draw [dashed] (0,0) -- (3,0);
\filldraw (3,2) circle (1pt) node[anchor = west]{$I^+$};

\end{tikzpicture}
\end{subfigure}

\caption{The red curves marks where hyperbolicity is lost, and the shaded region indicates the corresponding domain.}
\label{LossofHyperbolicity}
\end{figure}
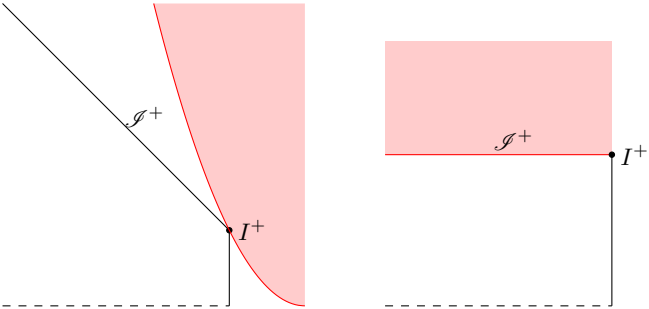
In the diagonal representation, the region where hyperbolicity is lost is bounded by a parabola outside $\mathscr{I}$, while in the horizontal representation, its boundary coincides with $\mathscr{I}$. However, our goal is to extract data on $\mathscr{I}^+$. 
As observed in \cite{GlobalSimulationsJoerg}, data can be evolved arbitrarily close to $\mathscr{I}^+$, but as this limit is approached, the time step required for convergence becomes prohibitively small. We therefore adopt the diagonal representation, which allows us to extract data exactly on $\mathscr{I}^+$ while avoiding the region where hyperbolicity is lost and maintaining stable evolution.

\subsection{The equations and issues to overcome}
The main difficulties occur at the top and bottom of the cylinder during the evolution. As discussed above, the evolution equations lose hyperbolicity there, as well as in a region beyond $\mathscr{I}$. In particular, the evolution equation for $\psi_0$ ($\psi_4$) cannot be evaluated at the bottom (top) of the cylinder. As a result, our standard approach on the cylinder, based on explicit time-stepping methods such as RK4, breaks down at these locations. Additionally, in the diagonal representation, $\mathscr{I}^-$ is not a level set of $t$, which is incompatible with evolution of spacelike hypersurfaces parametrised by these level sets.

The difficulties at the bottom of the cylinder can easily be addressed by applying an implicit numerical method, such as implicit Euler, for the first time step. This avoids the need to evaluate the evolution equations directly at $I^-$, allowing the data to be evolved along the cylinder to $I^+$. The success of this approach also suggests a way to overcome the second issue. In the horizontal representation, $\mathscr{I}^-$ is a level set of $t$, so the initial surface for the evolution can naturally be chosen as $\scrim$. In this representation, however, the equations lose hyperbolicity everywhere on $\mathscr{I}^-$ rather than only at the cylinder. Nevertheless, the same implicit method can be applied everywhere on $\scrim$, not just at $I^-$.

Although the above resolves the issues at the bottom of the cylinder, the implicit approach does not allow us to evolve beyond the top of the cylinder $I^+$. To proceed, we introduce a novel approach that exploits the geometry of the diagonal representation. We observe that the problematic $\psi_4$ evolution equation contains both a time and a radial derivative. This allows us to rewrite it as an evolution equation in the radial direction (see Eq.~\eqref{ethevo}),
\begin{equation}\label{eq:psi4r}
    \partial_r\psi_4 = \frac{t\kappa' - 1}{\kappa}\partial_t\psi_4 + \frac{1}{1-r}a_3\psi_3 + \Big{(}3\frac{\kappa'}{\kappa} + \frac{1}{1-r}\Big{)}\psi_4.
\end{equation}

This equation is regular away from $r=1$ and will only be used in the region $r<1$. Eq.~\eqref{eq:psi4r} is written in a form that makes explicit how $\psi_4$ evolves in the radial direction. Since the equations for the remaining four components are regular at $I^+$, they can be evolved using standard numerical methods. Moreover, sufficiently far along $\scrip$, $\psi_4$ can also be evolved without difficulty.
As a result, all five $\psi_k$ are known in the green regions of a given time slice shown in Fig.~\ref{fig:PastIpicture}. In the blue regions, only $\psi_4$ is unknown and can be obtained by integrating Eq.~\eqref{eq:psi4r}. This procedure can be iterated to move sufficiently far away from $\mathscr{I}^+$ to perform a regridding process that removes a piece of the outermost part of the computational domain. This in turn allows the evolution to be extended arbitrarily far into the future.
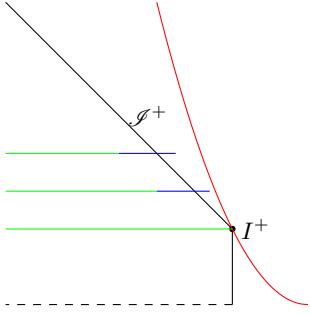
\begin{figure}[htbp]
\centering
\begin{tikzpicture}[xscale=-1]

\draw (0,1)  -- node[anchor = west]{$\mathscr{I}^+$} (3,4);
\draw (0,0) -- (0,1);
\draw [dashed] (0,0) -- (3,0);
\filldraw (0,1) circle (1pt) node[anchor = west]{$I^+$};

\draw[red] plot[smooth,domain=-1:1] (\x, {(1+\x)^2});

\draw [green] (0,1) -- (3,1);
\draw [green] (1,1.5) -- (3,1.5);
\draw [green] (1.5,2) -- (3,2);

\draw [blue] (0.3,1.5) -- (1,1.5);
\draw [blue] (0.75,2) -- (1.5,2);

\end{tikzpicture}
\caption{Schematic depiction of the method to evolve beyond $I^+$. The horizontal lines are level sets of $t$ with green indicating where $\psi_4$ has been evaluated by time evolution and blue with spatial evolution via Eq.~\eqref{eq:psi4r}.}
\label{fig:PastIpicture}
\end{figure}

\subsection{Numerical implementation}\label{sec:numericalimpltimeev}
% 100 iterations

The numerical evolution from $\scrim$ proceeds as follows. We start with the horizontal representation so that the initial surface is a level set of $t$, and employ an implicit Euler method to advance the solution for a fixed number of time steps. We then switch to an explicit RK4 scheme for further evolution. Once $t=0$ is reached, the coordinates of the horizontal and diagonal representations coincide, and we transform the data to the diagonal representation. The evolution is then continued in this representation until $I^+$ is reached. Spatial derivatives are approximated using a fourth-order finite-difference operator satisfying the summation-by-parts property \cite{osti_7184523}. Upon reaching the top of the cylinder, the evolution proceeds as illustrated in Fig.~\ref{fig:PastIpicture}:

\begin{itemize}
\item[1.] Evolve $\psi_0,\psi_1,\psi_2,$ and $\psi_3$ forward one time step everywhere on the time slice.
\item[2.] Evolve $\psi_4$ forward one time step in a region sufficiently far away from the cylinder.
\item[3.] Using the computed value of $\psi_4$ as an initial condition, together with the remaining $\psi_k$ from step 1, integrate $\psi_4$ radially outwards.
\end{itemize}

In order to integrate Eq.~\eqref{eq:psi4r}, we require not only the values of the $\psi_k$ components, but also the temporal derivative of $\psi_4$. This is approximated using backward differentiation formulas, which estimate time derivatives from previous time slices. To perform the radial integration, we use the fourth-order Adams--Bashforth method \cite{Adams_Bashforth_1883, Butcher_Linear_Multistep}, a linear multistep method.
This choice is motivated by the fact that Runge--Kutta methods of order higher than two require function evaluations at intermediate points, i.e.\ between grid points. In contrast, the Adams--Bashforth method uses values at previous grid points to achieve higher-order accuracy. Since the standard evolution proceeds without difficulty in the interior of the spacetime, we have access to sufficiently many interior points for this purpose.

Using this approach, we are able to evolve solutions from past null infinity, through the cylinder, and extract data at future null infinity.

\subsection{Evolving an exact solution}
Evolving an exact solution to the linearised GCFE provides a useful test of our numerical methods, as it allows direct comparison between numerical and exact results. This enables us to assess the accuracy of the scheme.

The exact solution with $l=2$, $m=0$ derived in \cite{DoulisPHD} is given by
\begin{equation}\label{ExactSoln}
\begin{split}
    \psi_0 &= \frac{\kappa^3(1-r+t\kappa)^4}{(1- r)^5},\\
    \psi_1 &= -\frac{2\kappa^3(1 - r-t\kappa)(1 - r+t\kappa)^3}{(1-r)^5},\\
    \psi_2 &= \frac{\sqrt{6}\kappa^3(1-r-t\kappa)^2(1-r+t\kappa)^2}{(1-r)^5},\\
    \psi_3 &= -\frac{2\kappa^3(1-r-t\kappa)^3(1-r+t\kappa)}{(1-r)^5},\\
    \psi_4 &= \frac{\kappa^3(1 - r-t\kappa)^4}{(1-r)^5}.
\end{split}
\end{equation}
By appropriately choosing $\kappa$, this solution can easily be expressed in the diagonal or horizontal representations. Initial and boundary data are then obtained by evaluating the exact solution at the desired space and time coordinates.

To test the convergence of our code, we evolve the exact solution at different resolutions and verify that the difference between numerical and exact solutions decreases with the expected order of convergence. The convergence plots in Fig.~\ref{exactsolnconvergencescri} confirm that the numerical solution converges to the exact solution on $\mathscr{I}^+$ when evolving from $\scrim$ using the methods described in this section.

\begin{figure}[htbp]
    \includegraphics[width = 0.5 \textwidth]{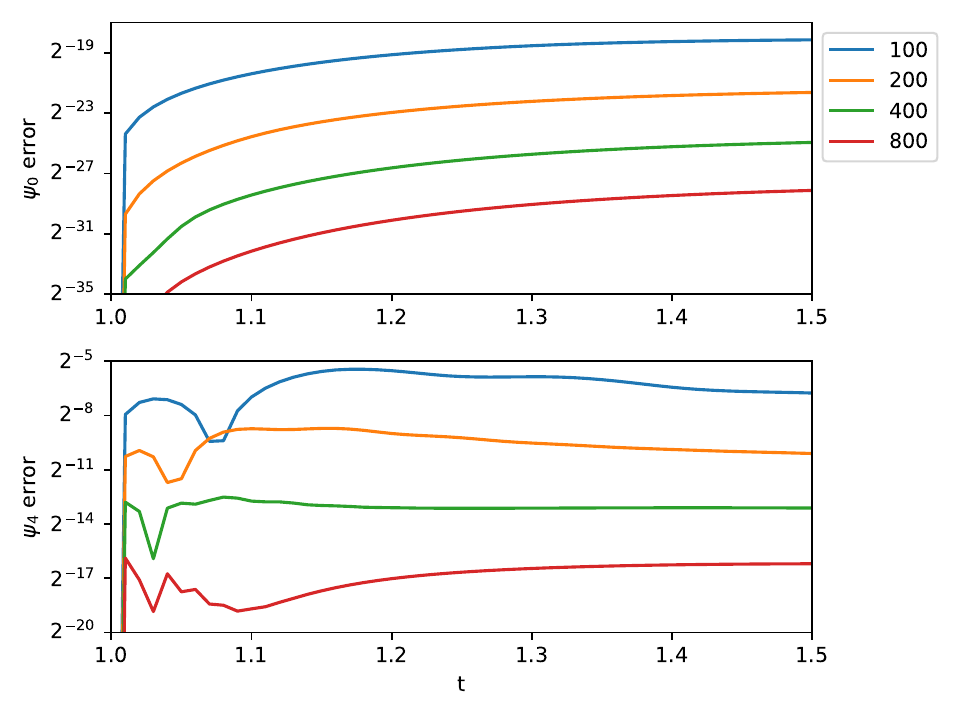}
    \caption{Convergence of $\psi_0$ and $\psi_4$ to the exact solution along $\scrip$ with increasing number of radial grid points.}
    \label{exactsolnconvergencescri}
\end{figure}

\subsection{A full numerical evolution of the linearised field equations}
Having validated our implementation, we now consider the evolution of a more general initial data set from $\scrim$. Using the methods described in Sec.~\ref{initial_data_numeric}, we construct data on $\mathscr{I}^-$ and evolve them to $\mathscr{I}^+$.

\subsubsection{Boundary conditions}
In contrast to the exact solution, the corresponding function values cannot be directly imposed as boundary conditions. However, the constraint propagation system contains no derivatives transverse to the level sets of $r$ \cite{kroon2023}. As a consequence, the ingoing characteristic mode remains unconstrained.
With respect to the computational domain, the characteristic modes at the left boundary correspond to $\psi_4$ (ingoing) and $\psi_0$ (outgoing). For simplicity, we impose fully reflective boundary conditions of the form
\begin{equation}\label{eq:reflective}
\psi_0 = -\overline{\psi_4},
\end{equation}
which is a special case of \emph{maximally dissipative boundary conditions}, see for example \cite{friedrich1999initial}.

\subsubsection{Convergence tests}
In order to test the robustness of the convergence of this method for non-exact initial data, we evolve numerically generated initial data (constructed as described in Sec.~\ref{initial data}) from $\mathscr{I}^-$ to $\mathscr{I}^+$ and compute the constraint violations for simulations at different resolutions.

In particular, we present convergence plots for an evolution of the data set defined in \cref{eq:fc two bump}. This choice is motivated by the fact that it exhibits non-trivial values of several components of the spin-2 field on the cylinder. This serves two purposes: firstly, it represents a more challenging test case, as several aspects of the method become trivial when certain $\psi_k$ components vanish identically; secondly, it provides a non-trivial check of the regularity conditions derived in Sec.~\ref{Regularity SCG}.

\begin{figure}[htbp]
\includegraphics[width = 0.5\textwidth]{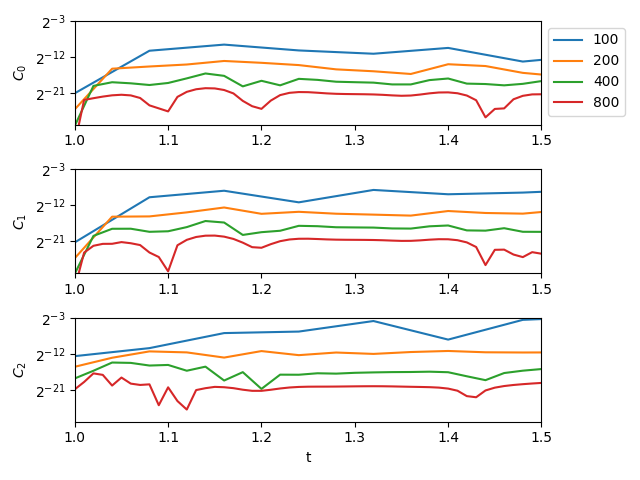}
\caption{Convergence of constraint quantities on $\mathscr{I}^+$ with increasing number of radial grid points.} \label{numerical_convergence}
\end{figure}

The convergence plots in Fig.~\ref{numerical_convergence} show that our method also converges for numerically generated initial data. As expected, the errors are larger than in the exact case, but a clear convergent behaviour is evident.

%\subsection{Evolving initial data sets}
\subsection{Scattering behaviour of gravitational wave data}
Having passed the basic convergence tests, we now evolve a solution that more closely resembles a physical gravitational wave, in order to check consistency with physical intuition.

An exact initial data set for an $l=2$, $m=0$ mode with compact support in $1<\hat r<2$ (corresponding to $r\in(-1,0)$) is constructed in the horizontal gauge by prescribing $\psi_4(\hat r)= \frac{1}{100} \sin(\pi \hat r)^6$ and solving for the remaining components $\psi_k$. This yields the initial data
\begin{align*}
    \psi_0(\hr) &= \frac{1}{400} \pi ^3\hr^3 \sin ^2(\pi \hr) [11 \pi \hr+4 \sin (2 \pi \hr)+18 \sin (4 \pi \hr)
    \\
    &\qquad +22 \pi \hr \cos (2 \pi \hr)+27
   \pi \hr \cos (4 \pi \hr)],
    \\
    \psi_1(\hr) &= -\frac{1}{200} \pi ^3\hr^3 \sin ^3(\pi \hr) [11 \cos (\pi \hr)+9 \cos (3 \pi \hr)],
    \\
    \psi_2(\hr) &= \frac{1}{200 \sqrt{6}}\sin ^4(\pi \hr) [12 \pi ^2\hr^2+\left(18 \pi ^2\hr^2-1\right) \cos (2 \pi \hr)
    \\
    &\qquad -6 \pi \hr \sin (2 \pi\hr)+1],
   \\
    \psi_3(\hr) &= \frac{1}{100} \sin ^5(\pi \hr) [\sin (\pi \hr)-3 \pi \hr \cos (\pi \hr)],
    \\
    \psi_4(\hr) &= \frac{1}{100}\sin^6(\pi \hr),
\end{align*}
% \begin{align*}
%     \psi_0(\hat r) &= \frac{1}{3} \pi ^3 \hr^3 (-2 \sin (2 \pi  \hat r)+4 \sin (4 \pi  \hat r)-\pi  \hat r \cos (2 \pi  \hat r)
%     \\ &\qquad +4 \pi  \hat r \cos (4 \pi  \hat r)),
%     \\
%     \psi_1(\hat r) &= \frac{1}{3} \pi ^3 \hat r^3 (\sin (2 \pi  \hat r)-2 \sin (4 \pi  \hat r)),
%     \\
%     \psi_2(\hat r) &= \frac{1}{2 \sqrt{6}}\sin (\pi  \hat r) ^2\left(4 \pi ^2 \hat r^2+\big(8 \pi ^2 \hat r^2-1\right) \cos (2 \pi  \hat r)
%     \\ &\qquad -4 \pi \hat  r \sin (2 \pi \hat r)+1\big),
%     \\
%     \psi_3(\hat r) &= \sin(\pi  \hat r) ^3 (\sin (\pi \hat  r)-2 \pi \hat  r \cos (\pi \hat  r)),
%     \\
%     \psi_4(\hat r) &= \sin (\pi  \hr)^4,
% \end{align*}
on $(1,2)$ and vanishing elsewhere.
\begin{figure}
\includegraphics[width = 0.5\textwidth]{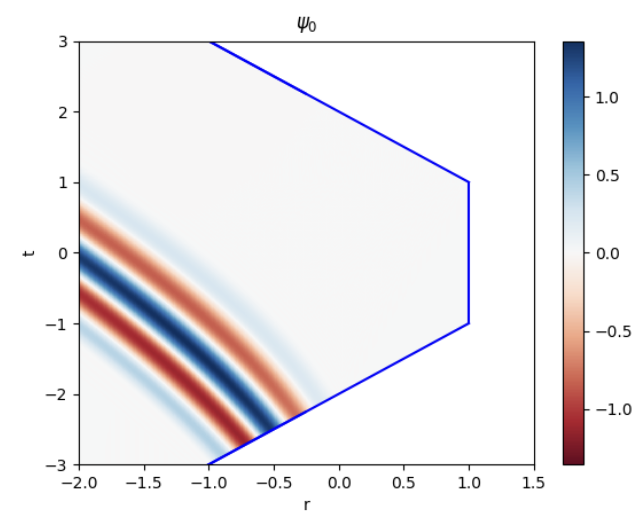}
\includegraphics[width = 0.5\textwidth]{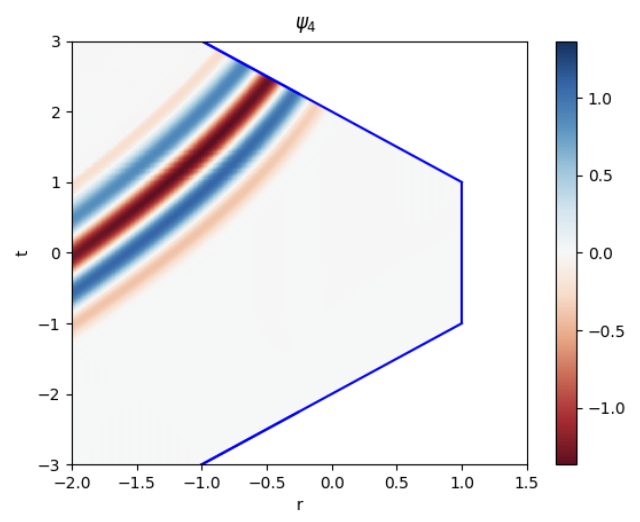}
\caption{Ingoing and outgoing components of the evolved gravitational wave in the diagonal gauge.}\label{fullevomerlyndata}
\end{figure}
As shown in Fig.~\ref{fullevomerlyndata}, the gravitational wave propagates inwards from $\mathscr{I}^-$, reflects off the boundary at $\hr=3$ ($r=-2$), where we have imposed the reflective boundary condition \eqref{eq:reflective}, and subsequently propagates to $\mathscr{I}^+$. This behaviour is consistent with the expected dynamics of a scattering problem.

We can also use the evolution system to investigate the piecewise polynomial initial data described in Sec.~\ref{sec:reg conds2}. We again use reflective boundary conditions and are able to achieve a stable evolution from past null infinity to future null infinity, including spatial infinity. The results of the evolution are depicted in Fig.~\ref{fig: piecewise polynomial evolution}.

\begin{figure}
    \includegraphics[width = 0.5\textwidth]{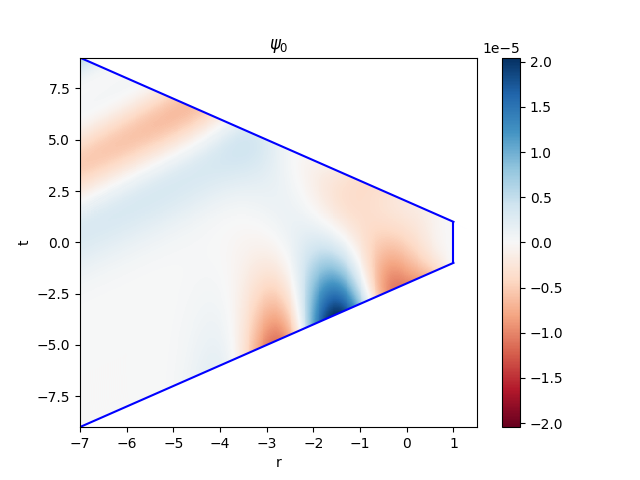}
    \includegraphics[width = 0.5\textwidth]{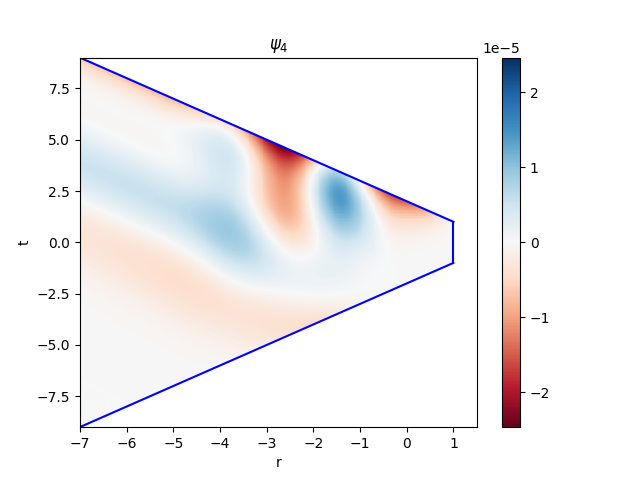}
\caption{Ingoing and outgoing components of the evolved piecewise polynomial data} \label{fig: piecewise polynomial evolution}
\end{figure}
We believe these to be the first numerical evolutions through spatial infinity in which data are prescribed explicitly on past null infinity and extracted on future null infinity.

\section{Summary and discussion}\label{sec:summary}
In this work, we have developed a framework for the construction and evolution of initial data for the linearised spin-2 system on conformally compactified Minkowski spacetime, with data prescribed on past null infinity $\scrim$ and evolved globally to future null infinity $\scrip$. Our results provide, to our knowledge, the first systematic analysis of initial data sets posed across the entirety of $\scrim$, extending from past timelike infinity $i^-$ to $I^-$, the bottom of the cylinder at spatial infinity.

A central outcome of this analysis is the identification of explicit and constructive conditions on the freely specifiable ingoing radiation field $\psi_0$ that ensure regularity of the full set of spin-2 components. In particular, we have obtained necessary and sufficient integral conditions for regular initial data, making precise how the behaviour of $\psi_0$ determines the global regularity properties of the solution. These results highlight a fundamental structural feature of the problem: Because the linearisation of rescaled Weyl spinor is governed by a first-order system intrinsic to $\scrim$, we cannot independently control the behaviour of the fields at both ends of null infinity. Regularity at $i^-$ and at the cylinder are therefore coupled in a nontrivial manner, and generic data give rise to polyhomogeneous expansions and logarithmic singularities. In this sense, the construction of globally regular initial data on $\scrim$ is significantly more constrained than might be anticipated from heuristic considerations.

An interesting consequence of these constraints is that physically intuitive choices of initial data, such as sharply localised or delta-function-like wave profiles, do not in general lead to globally regular solutions. In particular, while plane wave and impulsive wave solutions play an important role as exact solutions in general relativity, our results show that analogous constructions on $\scrim$ are incompatible with the regularity conditions required for smooth conformal evolution. This highlights a distinction between idealised exact solutions and admissible asymptotic initial data in the conformal setting.

On the constructive side, we have introduced a numerical procedure for generating complete initial data sets on $\scrim$ from a prescribed $\psi_0$. It implements the hierarchy of intrinsic equations in a stable and efficient manner, producing all remaining spinor components consistent with the constraints. Importantly, the construction relies only on the intrinsic transport system on null infinity and does not depend in an essential way on the linearisation. This suggests that it can, in principle, be extended to the fully nonlinear generalised conformal field equations when coupled to the remaining variables, providing a potential pathway towards constructing asymptotic initial data in the nonlinear regime.

We have also carried out numerical evolutions of the resulting data in both semi- and fully-compactified gauges. These simulations demonstrate, for the first time, stable time evolutions from $\scrim$ through the cylinder at spatial infinity and onward to $\scrip$. A key ingredient in achieving this is the development of new numerical treatments for the degeneracies that arise at the cylinder, allowing the system to be evolved across $I$ without loss of stability. This provides a concrete realisation of a fully global scattering problem in a time-evolution framework, complementing previous approaches based on hyperboloidal initial data.

There are several important directions for future work. One natural extension is the inclusion of the physical origin, which is part of our initial data investigations, but excluded in the time-evolution section. This step is essential for modelling more general configurations. Thinking toward the scattering framework within the fully non-linear GCFE, the use of a conformal Gau\ss ian system is particularly convenient, as it simplifies the GCFE significantly. Since the present fully-compactified gauges are not of this form, a global conformal Gau\ss ian system suitable for covering (at least) the region of the conformal space-time associated with the physical space-time and its conformal boundary must ultimately be formulated. This could be achieved naturally using the Einstein cylinder, whose spatial slices have topology $S^3$, although this requires numerical methods adapted to $S^3$ rather than those used here or in existing GCFE implementations \cite{Beyer_2017}. Further, it would require handling $i^0$ as a point, which may not be possible. Developing a conformal Gau\ss ian system that remains regular across the conformal manifold, including both timelike infinity and the cylinder at spatial infinity, is an important step toward tackling the full scattering problem and will ultimately require alterations to the algorithm we use in Sec.~\ref{sec:numericalimpltimeev}.

More broadly, the framework developed here provides a foundation for studying global properties of gravitational radiation using asymptotic initial data. Extending these ideas to the fully nonlinear regime offers the prospect of a complete scattering picture for asymptotically flat spacetimes within the conformal approach.

\begin{acknowledgments}
This work was supported by the Marsden Fund Council from government funding managed by the Royal Society Te Apārangi of New Zealand. The University of Canterbury's RCH cluster was used for numerical simulations. CS would like to thank Florian Beyer and Juan Valiente Kroon for useful discussions.
\end{acknowledgments}

\appendix

\section{Appendices}

% \nocite{*}

\subsection{Piecewise polynomial equations}
The initial data set using piecewise polynomials for Sec.~\ref{sec:reg conds2} are given by
\begin{widetext}
\allowdisplaybreaks
\begin{align}
\begin{split}
\tilde\psi_0(\hr)  &= \frac{\hr^3}{100}\chi_0(\hr) \quad \text{ where } 
\\
\chi_0(\hr) &= \begin{cases}
-(\hr-1)^2(\hr-3)^2, & 1\le \hr\le 2
\\
 3\hr^4-40\hr^3+186\hr^2-360\hr+247, & 2\le \hr\le 3
\\
-2\left (  \hr^{4}-24\hr^{3}+178\hr^{2}-528\hr+547\right ),& 3\le \hr\le 4
\\
-2\left (  \hr^{4}-8\hr^{3}-14\hr^{2}+240\hr-477\right ), & 4 \le \hr\le 5
\\
 3\hr^{4}-56\hr^{3}+378\hr^{2}-1080\hr+1079, & 5\le \hr\le 6
\\
-(\hr-5)^2(\hr-7)^2, & 6\le r\le 7
\end{cases},
\end{split} 
\label{eq:psi_0_example}
\\
\psi_0(\hr)  &= \frac{1}{(1+\hr)^3}\tilde\psi_0(\hr) = \frac{\hr^3}{100(1+\hr )^3} \chi_0(\hr)
\label{eq:psi_0_example_2}
\\
\begin{split} 
\psi_1(\hr) & = \frac{1}{42000\hr(1+\hr )^3}  \chi_1(\hr) \quad \text{ where } 
\\
\chi_1(\hr) &= \begin{cases}
105\hr^8 - 960\hr^7 +3080\hr^6 - 4032\hr^5 +1890\hr^4 - 83, &1\le \hr< 2
\\
- 3(105\hr^8 - 1600\hr^7+8680\hr^6-20160\hr^5+17290\hr^4-9871), & 2\le \hr< 3
\\
2(105\hr^8-2880\hr^7+24920\hr^6-88704\hr^5+114870\hr^4-320169), & 3\le \hr< 4
\\
10(21\hr^8 - 192 \hr^7-392\hr^6+8064\hr^5-20034\hr^4 + 250539 ),&4\le \hr< 5
\\
-5(63\hr^8 - 1344\hr^7+10584\hr^6-36288\hr^5+45318\hr^4 + 77047),& 5\le \hr< 6
\\
5(21\hr^8-576\hr^7+5992\hr^6-28224\hr^5+51450\hr^4-823543),& 6\le \hr\le 7
\end{cases}, 
\end{split}
\label{eq:psi_1_example}
\\
\begin{split} 
\psi_2(\hr) & = \frac{\sqrt{6}}{42000\hr(1+\hr )^3}  \chi_2(\hr) \quad \text{ where } 
\\
\chi_2(\hr) &=\begin{cases}
-15\hr^{8}-160\hr^{7}+616\hr^{6}-1008\hr^{5}+630\hr^{4}-176\hr+83, & 1\le \hr< 2
\\
45\hr^{8}-800\hr^{7}+5208\hr^{6}-15120\hr^{5}+17290\hr^{4}-32592\hr+29613 ,& 2\le \hr< 3
\\
-2(15\hr^{8}-480\hr^{7}+4984\hr^{6}-22176\hr^{5}+38290\hr^{4}-238368\hr+320169), &3\le \hr< 4
\\
-2(15\hr^{8}-160\hr^{7}-392\hr^{6}+10080\hr^{5}-33390\hr^{4}+679136\hr-1252695), & 4\le \hr< 5
\\
45\hr^{8}-1120\hr^{7}+10584\hr^{6}-45360\hr^{5}+75530\hr^{4}-108272\hr-385235, &5\le \hr< 6
\\
- 15\hr^{8}-480\hr^{7}+5992\hr^{6}-35280\hr^{5}+85750\hr^{4}-1882384\hr+4117715,&6\le \hr< 7
\end{cases},
\end{split} 
\label{eq:psi_2_example}
\\
\begin{split} 
\psi_3(\hr) & = \frac{1}{14000\hr(1+\hr )^3}  \chi_3(\hr) \quad \text{ where }
\\
\chi_3(\hr) & = \begin{cases}
5\hr^{8}-64\hr^{7}+308\hr^{6}-672\hr^{5}+630\hr^{4}-476\hr^{2}+352\hr-83, & 1\le \hr< 2
\\
-(15\hr^{8}-320\hr^{7}+2604\hr^{6}-10080\hr^{5}+17290\hr^{4}-46116\hr^{2}+65184\hr-29613),& 2\le \hr< 3
\\
2(5\hr^{8}-192\hr^{7}+2492\hr^{6}-14784\hr^{5}+38290\hr^{4}-227556\hr^{2}+476736\hr-320169), &3\le \hr< 4
\\
2(5\hr^{8}-64\hr^{7}-196\hr^{6}+6720\hr^{5}-33390\hr^{4}+460572\hr^{2}-1358272\hr+1252695), & 4\le \hr< 5
\\
-(15\hr^{8}-448\hr^{7}+5292\hr^{6}-30240\hr^{5}+75530\hr^{4}-308644\hr^{2}+216544\hr+385235), &5\le \hr< 6
\\
5\hr^{8}-192\hr^{7}+2996\hr^{6}-23520\hr^{5}+85750\hr^{4}-1142876\hr^{2}+3764768\hr-4117715,&6\le \hr< 7
\end{cases},
\end{split}
\label{eq:psi_3_example}
\\
\begin{split} 
\psi_4(\hr) & = \frac{1}{21000\hr(1+\hr )^3}  \chi_4(\hr) \quad \text{ where }
\\
\chi_4(\hr) & = \begin{cases}
-(3\hr^{8}-48\hr^{7}+308\hr^{6}-1008\hr^{5}+1890\hr^{4}-2128\hr^{3}+1428\hr^{2}-528\hr+83), & 1\le \hr< 2
\\
9\hr^{8}-240\hr^{7}+2604\hr^{6}-15120\hr^{5}+51870\hr^{4}-108976\hr^{3}+138348\hr^{2}-97776\hr+29613,& 2\le \hr< 3
\\
-2(3\hr^{8}-144\hr^{7}+2492\hr^{6}-22176\hr^{5}+114870\hr^{4}-361312\hr^{3}+682668\hr^{2}-715104\hr+ 320169), &3\le \hr< 4
\\
-2(3\hr^{8}-48\hr^{7}-196\hr^{6}+10080\hr^{5}-100170\hr^{4}+498848\hr^{3}-1381716\hr^{2}+2037408\hr-1252695), & 4\le \hr< 5
\\
9\hr^{8}-336\hr^{7}+5292\hr^{6}-45360\hr^{5}+226590\hr^{4}-647696\hr^{3}+925932\hr^{2}-324816\hr-385235, &5\le \hr< 6
\\
-(3\hr^{8}-144\hr^{7}+2996\hr^{6}-35280\hr^{5}+257250\hr^{4}-1190896\hr^{3}+3428628\hr^{2}-5647152\hr+4117715),&6\le \hr< 7
\end{cases}.
\end{split}
\label{eq:psi_4_example}
\end{align}
\end{widetext}

\bibliography{refs}% Produces the bibliography via BibTeX.

\end{document}